\documentclass{amsart}
\usepackage{fullpage}
\usepackage{newpxtext}
\usepackage{newpxmath}
\usepackage{hyperref}
\usepackage{cleveref}
\usepackage{microtype}%
\usepackage{algorithm, algorithmicx, algpseudocode}
\usepackage{amsfonts}
\usepackage{mathrsfs}
\usepackage{amssymb}
\usepackage{float}
\usepackage{pstricks,pst-node}
\usepackage[colorinlistoftodos,bordercolor=orange,backgroundcolor=orange!20,linecolor=orange,textsize=scriptsize,disable]{todonotes}
\title[Entropy games]{The operator approach to entropy games}
\thanks{The authors were partially supported by the ANR through the MALTHY INS project, and by the Gaspard Monge corporate sponsorship Program (PGMO) of 
EDF, Orange, Thales and Fondation Math\'ematique Jacques Hadmard.}

\author{Marianne Akian}
\address{INRIA and CMAP, \'Ecole polytechnique, CNRS, France, 
  \texttt{Marianne.Akian@inria.fr}}
\author{St\'ephane Gaubert}
\address{INRIA and CMAP, \'Ecole polytechnique, CNRS, France, 
  \texttt{Stephane.Gaubert@inria.fr}}
\author{Julien Grand-Cl\'ement}
\address{\'Ecole polytechnique, France\\
  \texttt{Julien.Grand-Clement@polytechnique.edu}}
\author{J\'er\'emie Guillaud}
\address{\'Ecole polytechnique, France\\
  \texttt{Jeremie.Guillaud@polytechnique.edu}}


\newcommand{\valueik}[2]{v_{\!#1}^{#2}}
\newcommand{\valuek}[1]{v^{#1}}

\newcommand{\valued}{{}^{\alpha}v}
\newcommand{\bitsize}[1]{\langle #1\rangle}
\newcommand{\ceil}[1]{\lceil #1 \rceil}
\newcommand{\R}{\mathbb{R}}

\newcommand{\cP}{\mathcal{P}}

\newcommand{\sK}{\mathscr{K}}

\newcommand{\A}{\mathcal{A}}
\newcommand{\B}{\mathcal{B}}

\newcommand{\E}{\mathcal{E}}
\newcommand{\M}{\mathcal{M}}

\newcommand{\Ftau}{{}^{\tau}F}
\newcommand{\Mt}[1]{{}^{#1}\!M}
\newcommand{\Mdt}[2]{{}^{#2}\!M^{#1}}
\newcommand{\Fdt}[2]{{}^{#2}F^{#1}}
\newcommand{\Ftauk}{{}^{\tau^k}F}
\newcommand{\Ftaup}{{}^{\tau'}F}
 \newcommand{\gammage}{\E}
 
 \newcommand{\gammagekdbar}{\E^k_{\bar{d}}}
 
 \newcommand{\gammageinftydbar}{\E^{\infty}_{\bar{d}}}
 
 \newcommand{\gammae}{\E^{\text{orig}}}
 \newcommand{\gammaeinfty}{\E^{\text{orig},\infty}}

 \newcommand{\gammampg}{\mathcal{K}\mathcal{L}}
 \newcommand{\gammampgd}{^{\alpha}\gammampg}
 \newcommand{\gammampginfty}{\gammampg^{\infty}}
 
\newcommand{\disfac}[1]{\gamma(#1)}

\newcommand{\argmax}{\mathop{\mathrm{argmax}}}
\newcommand{\argmin}{\mathop{\mathrm{argmin}}}

\newcommand{\etasep}{\eta_{\text{sep}}}
\newcommand{\maxweight}{W}
\newtheorem{theorem}{Theorem}
\newtheorem{proposition}[theorem]{Proposition}
\newtheorem{lemma}[theorem]{Lemma}
\newtheorem{fact}[theorem]{Fact}
\newtheorem{assumption}[theorem]{Assumption}
\newtheorem{corollary}[theorem]{Corollary}

\theoremstyle{remark}
\newtheorem{example}{Example}
\newtheorem{remark}{Remark}
\theoremstyle{definition}
\newtheorem{definition}{Definition}

\renewcommand{\geq}{\geqslant}
\renewcommand{\leq}{\leqslant}

\newcommand{\unit}{{e}}
\begin{document}

\maketitle

\begin{abstract}
Entropy games and matrix multiplication games have been recently
introduced by Asarin et al. They model the situation in which one
player (Despot) wishes to minimize the growth rate of a matrix
product, whereas the other player (Tribune) wishes to maximize it. We
develop an operator approach to entropy games.  This allows us to show
that entropy games can be cast as stochastic mean payoff games in
which some action spaces are simplices and payments are given by a
relative entropy (Kullback-Leibler divergence).  In this way, we show
that entropy games with a fixed number of states belonging to Despot
can be solved in polynomial time.  This approach also allows
us to solve these games by a policy iteration algorithm, which
we compare with the spectral simplex algorithm developed by Protasov.
\end{abstract}

\section{Introduction} 
\subsection{Entropy games and matrix multiplication games}
Entropy games have been introduced by Asarin et al.~\cite{asarin}.
They model the situation in which two players with conflicting interests,
called ``Despot'' and ``Tribune'', wish 
to minimize or to maximize a topological entropy representing the freedom
of a half-player, ``People''. 
Entropy games are special ``matrix multiplication games'', in which two players alternatively 
choose matrices in certain
prescribed sets; the first player wishes to minimize 
the growth rate of the infinite matrix product obtained
in this way, whereas the second player wishes to maximize it.
Whereas matrix multiplication games are hard in general
(computing joint spectral radii is a special case),
entropy games correspond to a tractable subclass
of multiplication games, in which the matrix sets
have the property of being invariant by row interchange,
the so called independent row uncertainty (IRU) assumption,
sometimes also called \textit{row-wise} or \textit{rectangularity} assumption.
In particular, Asarin et al.\ showed in~\cite{asarin}
that the problem of comparing the value of an entropy game to a given rational
number is in NP $\cap$ coNP, giving to entropy games a status
somehow comparable to other important classes of games with an unsettled
complexity, including mean payoff games, simple stochastic games,
or stochastic mean payoff games, see~\cite{andersson_miltersen}
for background.  

Another motivation to study entropy games arises from risk sensitive control~\cite{FHH97,FHH99,anantharam}: as we shall see, essentially the same
class of operators arise in the latter setting. 
A recent application of entropy games to the approximation
of the joint spectral radius of
nonnegative matrices (without making the IRU 
assumption) can be found in~\cite{stott}.
Other motivations originate from symbolic dynamics~\cite[Chapter 1.8.4]{lothaire}.

\subsection{Contribution}
We first show that entropy games, which were introduced
as a new class of games, 
are equivalent to a class of zero-sum mean payoff stochastic games with perfect information, in which some action spaces are simplices, and the instantaneous payments are given by a Kullback-Leibler entropy. Hence, entropy games
fit in a classical class of games, with a ``nice''
payment function over infinite action spaces.

To do so, we introduce a slightly more expressive variant
of the model of Asarin et al~\cite{asarin}, in which
the initial state is prescribed (the initial
state is chosen by a half-player, People, in the original model).
This may look like a relatively minor extension, so we keep the name ``entropy game'' for it, but this extension is essential to develop an operator
approach and derive consequences from it.
We show that the main results known for stochastic mean payoff
games with finite actions space and perfect information,
namely the existence of the
value and the existence of optimal positional strategies,
are still valid for entropy games
(Theorems~\ref{th-1} and~\ref{cor-morphism}). This is derived 
from a model theory approach of Bolte, Gaubert, and Vigeral~\cite{bolte2013},
together with the observation that the dynamic programming
operators of entropy games are definable in the real exponential field.
Then, a key ingredient is the proof of existence of Blackwell optimal policies,
as a consequence of o-minimality, see \Cref{cor-ominimal}.
Another consequence of the operator approach is the existence
of Collatz-Wielandt optimality certificates for entropy games, Theorem~\ref{th-cw}. When specialized to the one player case, this leads to a convex programming
characterization of the value, Corollary~\ref{cor-cw},
which can also be recovered from a characterization of Anantharam and Borkar~\cite{anantharam}.

Our
main result, Theorem~\ref{cor-polytime}, shows that entropy games
in which Despot has a fixed number of significant states (states with
a nontrivial choice) can be solved {\em strategically} in polynomial time,
meaning that optimal (stationary) strategies can be found in polynomial time.
Thus, entropy games are somehow similar to stochastic mean payoff games, for
which an analogous fixed-parameter tractability result holds (by reducing
the one player case to a linear program).
This approach also reveals
a fundamental asymmetry between the players Despot and Tribune:
our approach does not lead to a polynomial bound if one fixes
the number of states of Tribune.
In our proof, o-minimality arguments allow
a reduction from the two-player to the one-player case
(\Cref{cor-morphism}).
Then, the one-player case is dealt with using several
ingredients:  ellipsoid method, separation
bounds between algebraic numbers, and
results from Perron-Frobenius theory. 

The operator approach also allows one to obtain practically
efficient algorithms to solve entropy games. In this way, the
classical policy iteration of Hoffman-Karp~\cite{HoffmanKarp}
can be adapted to entropy games. We report experiments showing that when specialized to one player problems, policy iteration yields a speedup by one
order of magnitude by comparison with 
the ``spectral simplex'' method recently introduced
by Protasov~\cite{protasov}.

Let us finally complete the discussion of related works. 
The formulation of entropy games in terms
of ``classical'' mean payoff games in which the payments
are given by a Kullback-Leibler entropy builds on 
known principles in risk sensitive control~\cite{FHH99,anantharam}.
It can be thought
as a version for two player problems 
of the Donsker-Varadhan characterization of the Perron-eigenvalue~\cite{donsker}. The latter is closely related to the log-convexity property of the spectral
radius established by Kingman~\cite{kingman}.
A Donsker-Varadhan type formula for risk sensitive problems, which can be applied in particular to Despot-free player entropy games, has been recently
obtained by Anantharam and Borkar, in a wider setting allowing an infinite state space~\cite{anantharam}. In a nutshell, for Despot-free
problems, the Donsker-Varadhan formula appears to be the (convex-analytic) dual of the Collatz-Wielandt formula. 
Chen and Han~\cite{chen} developed
a related convex programming
approach to solve the entropy maximization problem for Markov chains 
with uncertain parameters. 
We also note that the present Collatz-Wielandt
approach, building on~\cite{agn}, yields 
an alternative to the approach of~\cite{asarin}
using the ``hourglass alternative'' of~\cite{kozyakin}
to produce concise certificates allowing one to bound
the value of entropy games.  By comparison with~\cite{asarin},
a essential difference is the use of o-minimality arguments:
these are needed because we study the more precise version
of the game, in which the initial state is fixed. 
Indeed, a counter example of Vigeral shows that the mean payoff
may not exist in such cases without an
o-minimality assumption~\cite{Vipreprint},
whereas the existence of the mean payoff holds universally (without restrictions
of an algebraic nature on the Shapley operator) if
one allows one player to choose the initial state, 
see e.g. Proposition~2.12 of~\cite{AGGut10}.
Finally, the identification of tractable subclasses of matrix multiplication
games can be traced back at least to the work of Blondel and Nesterov~\cite{blondel}.

\section{Entropy games}
\label{sec-genent}
An entropy game $\gammage$ is a perfect information game
played on a finite directed weighted graph $G$.
There are $2$ players, ``Despot'', ``Tribune'', and a half-player
with a nondeterministic behavior, ``People''.
The set of nodes of the graph is written as the disjoint union
$D\cup T \cup P$, where $D,T$ and $P$ represent sets of states in which Despot,
Tribune, and People play. We assume that the set of arcs $E$ is 
included in $(D\times T) \cup (T \times P) \cup (P\times D)$, meaning that 
Despot, Tribune, and People alternate their actions. A {\em weight}
$m_{pd}$,
which is a positive real number, is attached to every arc $(p,d)\in P\times D$.
All the other arcs in $E$ have weight $1$.
An initial state, $\bar{d}\in D$, is known to the players.
A token, initially in node $\bar{d}$,
is moved in the graph according to the following rule.
If the token is currently in a node $d$ belonging to $D$,
then, Despot chooses an arc $(d,t)\in E$ and moves the token
to a node $t$. 
Similarly, if the token is currently in a node 
$t\in T$, Tribune chooses an arc $(t,p)\in E$ and moves
the token to node $p$. Finally, if the token is in a node
$p\in P$, People chooses an arc $(p,d')\in E$ and moves
the token to a node $d'\in D$. We will assume that every
player has at least one possible action in each state
in which it is his or her turn to play. In other words,
for all $d\in D$, the set of actions $\{(d,t) \in E\}$ 
must be nonempty, and similar conditions apply to $t\in T$ and $p\in P$. 

A {\em history} of the game consists of a finite path in the directed graph $G$,
starting from the initial node $\bar{d}$. The {\em number of turns}
of this history is defined to be the length of this path, each arc
counting for a length of one third. 
The {\em weight} of a history
is defined to be the product of the weights of the arcs arising on this path.
For instance, a history
$(d_0,t_0,p_0,d_1,t_1,p_1,d_2,t_2)$ where
$d_i\in D$, $t_i\in T$ and $p_i\in P$, 
makes $2$ and $1/3$ turn, and its weight is $m_{p_0d_1}m_{p_1d_2}$.

A {\em strategy} of Player Despot is 
a map $\delta$ which assigns to every history 
ending in some node $d$ in $D$ 
an arc of the form $(d,t)\in E$.
Similarly,  a {\em strategy} of Player Tribune is 
a map $\tau$ which assigns 
an arc $(t,p)\in E$
to every history ending with a node $t$ in $T$.
The strategy $\delta$ is said to be {\em positional}
if it only depends on the last node $d$ which has been
visited and eventually of the number of turns.
Similarly, the strategy $\tau$ is 
said to be {\em positional} if it only depends on $t$
and eventually of the number of turns.
These strategies are in addition {\em stationary}, if they
do not depend on the number of turns.

For every integer $k$, we define as follows the {\em game in horizon $k$}
with initial state $\bar{d}$, $\gammagekdbar$. 
We assume that Despot and Tribune play according to the strategies
$\delta,\tau$. Then, People plays in a nondeterministic way. Therefore,
the pair of strategies $\delta,\tau$ allows for different
histories. The payment received by Tribune, in $k$ turns, 
is denoted by 
$R_{\bar{d}}^k(\delta,\tau)$.
It is defined as the sum of the weights of all the paths
of the directed graph $G$ of length $k$ with initial node $\bar{d}$
determined by the strategies $\delta$ and $\tau$:
each of these paths corresponds to different successive choices of People,
leading to different histories allowed by the strategies $\delta,\tau$. 
The payment received by Despot is 
defined to be the opposite of $R_{\bar{d}}^k(\delta,\tau)$,
so that the game in horizon $k$ is zero-sum.
In that way, the payment $R_{\bar{d}}^k$ measures the ``freedom'' of People,
Despot wishes to minimize it whereas Tribune wishes to maximize
it.

We say that the game $\gammagekdbar$
in horizon $k$ with initial state $\bar{d}$
{\em has the value} $V^k_{\bar{d}}$
if for all $\epsilon>0$, there is a strategy $\delta^*_\epsilon$
of Despot  such that for all strategies $\tau$ of Tribune,
\begin{align}
\epsilon+ V^k_{\bar{d}} 
\geq R^k_{\bar{d}}(\delta^*_\epsilon,\tau)  \enspace, 
\label{e-def-value0}
\end{align}
and similarly, there is a strategy $\tau^*_\epsilon$
of Tribune such that for all strategies $\delta$ of Despot,
\begin{align}
R^k_{\bar{d}}(\delta,\tau^*_\epsilon) \geq V^k_{\bar{d}}-\epsilon \enspace .
\label{e-def-value1}
\end{align}
The strategies $\delta^*_\epsilon$ and $\tau^*_\epsilon$ are said to be $\epsilon$-optimal. In other words, Despot can make sure his loss will not exceed
$V^{k}_{\bar{d}}+\epsilon$ by playing $\delta^*_\epsilon$, and Tribune can make sure
to win at least $V^{k}_{\bar{d}}-\epsilon$ by playing $\tau^*_\epsilon$.
The strategies $\delta^*$ and $\tau^*$ are optimal if they are $0$-optimal, i.e., if 
we have the saddle point property:
\begin{align}\label{e-def-value}
R^k_{\bar{d}}(\delta,\tau^*) \geq R^k_{\bar{d}}(\delta^*,\tau^*)
= V^k_{\bar{d}} 
\geq R^k_{\bar{d}}(\delta^*,\tau) \enspace ,
\end{align}
for all strategies $\delta,\tau$ of Despot and Tribune.
If the value $V^k_{\bar{d}}$ exists for all choices
of the initial state $\bar{d}$, we define the {\em value vector}
of the family of games $(\gammage^k_{d})_{d\in D}$ in horizon $k$,
to be $V^k:=(V^k_d)_{d\in D}\in \R^D$. 

We now define the {\em infinite horizon game} $\gammageinftydbar$, in which the payment received by Tribune is given by 
\[
R_{\bar{d}}^\infty(\delta,\tau):= 
\limsup_{k\to \infty} (R_{\bar{d}}^k(\delta,\tau))^{1/k}
\]
and the payment received by Despot is the opposite
of the latter payment. (The choice of limsup is somehow arbitrary,
we could choose liminf instead without affecting the results which follow.)
The {\em value} $V^\infty_{\bar{d}}$ of the infinite horizon game $\gammageinftydbar$,
and the optimal strategies in this game, are still defined by a saddle
point condition, as in~\eqref{e-def-value0}, \eqref{e-def-value1}, \eqref{e-def-value}, the payment 
$R^k_{\bar{d}}(\delta,\tau)$ being now replaced by 
$R_{\bar{d}}^\infty(\delta,\tau)$.

We denote by $V^\infty=(V^\infty_d)_{d\in D}\in \R^D$ the 
{\em value vector} of the infinite horizon games $(\gammage^\infty_{d})_{d\in D}$.

We associate to the latter games the dynamic programming operator $F: \R^D\to \R^D$, such that, for all $X\in \R^D$, and $d\in D$,
\begin{align}
F_d(X) = \min_{(d,t)\in E} \max_{(t,p)\in E} \sum_{(p,d') \in E} m_{pd'}X_{d'}\enspace .
\label{e-def-dp}
\end{align}
To relate this operator with the value of the above finite or infinite horizon
games, we shall interpret these games as zero-sum 
stochastic games with expected multiplicative criteria.
The one-player case was studied in particular
by Howard and Matheson under the name of
risk-sensitive Markov decision processes~\cite{Howard-Matheson}
and by Rothblum under the name of multiplicative Markov decision processes,
see for instance~\cite{rothblum}. %

For any node
$p\in P$, we denote by $E_p:= \{(p,d)\in E\}$ the set
of actions available to People in state $p$, and we denote by
$q_p$ the probability measure on $E_p$ obtained by normalizing the
restriction of the weight function $m$ to $E_p$:
$q_{pd}=m_{pd}/\disfac{p}$ with $\disfac{p}=\sum_{(p,d')\in E_p} m_{pd'}$.
Then, $F$ can be rewritten as 
\[ %
F_d(X) = \min_{(d,t)\in E} \max_{(t,p)\in E} 
\left (\disfac{p} \sum_{(p,d') \in E} q_{pd'}X_{d'}\right) \enspace .
\] %

A pair of strategies $\delta$ and $\tau$ of both players,
determine the stochastic process $(D_k,T_k,P_k)_{k\geq 0}$
with values in $D\times T\times P$, such that
$P(D_{k+1}=d'\mid H)=q_{pd'}$ for all $d'\in D$ 
and all histories $H$ having $k-1/3$ turns and ending in $p\in P$, and 
such that the transitions from
$D$ to $T$ and $T$ to $D$ are deterministicaly determined by the strategies
$\delta$ and $\tau$ respectively as in the above description of the entropy 
games $\gammage$.
Then, the payoff of the entropy game with horizon $k$ starting in $\bar{d}$,
$\gammagekdbar$, is equal to the following
expected multiplicative/ risk-sensitive criterion:
\[ R_{\bar{d}}^k(\delta, \tau)= \mathbb{E}\left( \disfac{P_0}\cdots \disfac{P_{k-1}}\mid D_0= \bar{d}\right)\enspace .\]

\begin{proposition}\label{prop-simple}
The value of the entropy game in horizon $k$ with initial state $d$,
$\gammage^k_d$,
does exists. The value vector $V^k$ of this game is determined by the relations
$V^0 =\unit$,  $V^k= F(V^{k-1})$, $k=1,2,\dots$,
where $\unit$ is the unit vector $(1,...,1)^{\top}$ of $\R^D$.
Moreover, there exist optimal strategies for Despot and Tribune that
are positional.
\end{proposition}
\begin{proof}
This result follows from a classical dynamic programming argument.
Indeed, in the one player case, that is when there is only one choice of
$\delta$ or one choice of $\tau$, that is when the operator $F$ contains
only a ``min'' or a ``max'', the game is in the class
of Markov Decision Problems with multiplicative criterion and 
the Dynamic Programming Principle has already been proved 
in this setting in~\cite{Howard-Matheson,rothblum},
see also~\cite[Th.~1.1, Chap 11]{whittle}.
This shows that the game has a value
which satifies $V^k=F(V^{k-1})$ and $V^0=e$, and that an optimal
strategy is obtained using these equations. For instance for a ``max''
(when Despot has only one choice), Tribune chooses any action 
$(t,p)$ attaining the maximum in
\[
 \max_{(t,p)\in E} \sum_{(p,d') \in E} m_{pd'}V^{k-1}_{d'} =
 \max_{(t,p)\in E}  \left (\disfac{p} \sum_{(p,d') \in E}q_{pd'}V^{k-1}_{d'} \right) \enspace .
\]
The resulting strategy $\tau^*$ is positional 
and it is optimal among all strategies
$\tau$. A similar result holds for a ``min'', leading
to a positional strategy $\delta^*$ for Despot.

Let us now consider the general two-player case.
Define the sequence of vectors $V^k$ by
\begin{align}
V^k_d = \min_{(d,t)\in E} 
\max_{(t,p)\in E} \sum_{(p,d') \in E} m_{pd'}V^{k-1}_{d'}\enspace .
\label{e-dp0}
\end{align}
with $V^0_d=1$, for all $d\in D$.
We construct candidate strategies $\delta^*$ and $\tau^*$, depending
on the current position and number of turns, as follows. 
In state $d$, if there remains $k$ turns to be played, Despot selects
an action $(d,t)$ achieving the minimum in~\eqref{e-dp0}.
We denote by $\delta^*(d,k)$ the value of $t$ such that $(d,t)$
is selected.
 In state $t$,
if there remains $k-1/3$ turns to be played, Tribune chooses any action
$(t,p)$ attaining the maximum in
\[
 \max_{(t,p)\in E} \sum_{(p,d') \in E} m_{pd'}V^{k-1}_{d'} \enspace .
\]
Now, if Player Despot plays according to $\delta^*$, we
obtain a reduced one player game. It follows from 
the same dynamic programming principle as above
(applied here to time dependent transition probabilities $q$
and factors $\disfac{\cdot}$)
that the value vector $V^{\delta^*,k}$ of this reduced
game in horizon $k$ does exist and satisfies the recursion
\[
V^{\delta^*,k}_d = \max_{(\delta^*(d,k),p)\in E} \left(\disfac{p}\sum_{(p,d') \in E} q_{pd'}V^{\delta^*,k-1}_{d'}\right)\enspace ,
\]
with $V^{\delta^*,0}_d=1$, for all $d\in D$. 
Since $V^{\delta^*,k}$ is the value, we have
$V^{\delta^*,k}_d\geq R_{d}^k(\delta^*, \tau)$
for all strategies $\tau$ of Tribune.
Noting that $V^{\delta^*,k}_d=V^{k}_d$
by definition of $\delta^*$, we deduce that Despot, by playing $\delta^*$,
can guarantee that his loss in the horizon $k$ game starting from state
$d$ will not exceed $V^k_d$. A dual argument
shows that by playing $\tau^*$, Tribune can guarantee that his
win will be at least $V^k_d$.
\end{proof}

\begin{example}\label{ex-1}
Consider the entropy game whose graph and dynamic programming operator
are given by:
\begin{center}
\begin{tabular}{cc}
\ \ \ \ \ \includegraphics[scale=0.07]{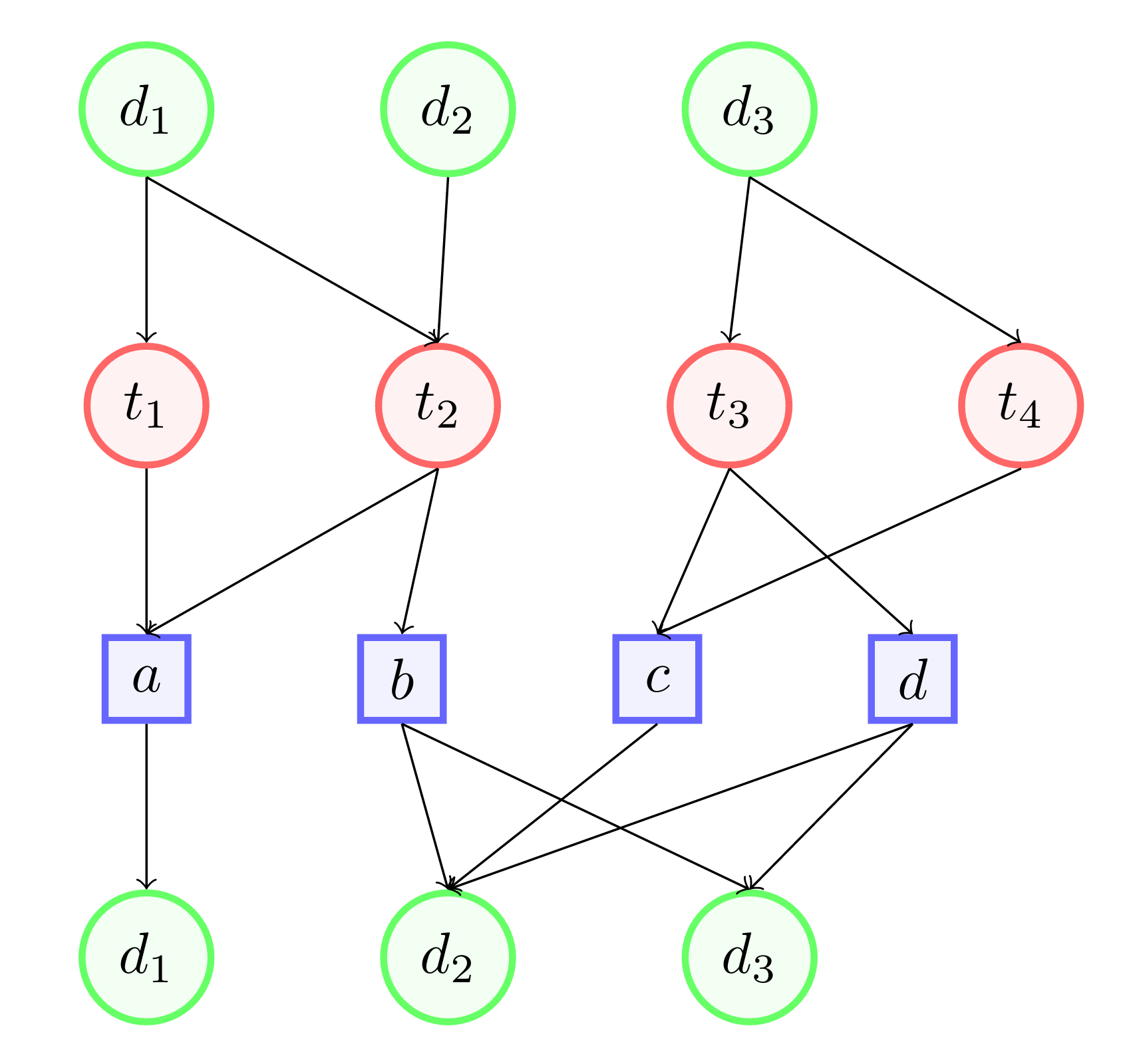}&
\ \ \ \ \ \ \ \ \ \ \begin{minipage}[b]{7cm}
$F_{1}(X) = \min\big(X_{1}, \max(X_{1},X_{2}+X_{3}) \big)$,\\ \ \\ 
$F_{2}(X) = \max\big(X_{1},X_{2}+X_{3} \big)$, \\ \ \\
$F_{3}(X) = \min \big( \max (X_{2}, X_{2} + X_{3} ) , X_2\big)$.\\ \ \\ 
\ \\
\end{minipage}
\end{tabular}
\end{center}
For readability, the states of Despot are shown twice on the picture.
Here, $D = \{ d_{1}, d_{2},d_{3} \}$, $T = \{t_{1},t_{2},t_{3},t_{4} \}$, $ P = \{a,b,c,d\}$, 
$E = \{ (d_{1},t_{1})$, $(d_{1},t_{2})$, $(d_{2},t_{2})$, $(d_{3},t_{3})$, $(d_{3},t_{4})$, $(t_{1},a)$, $(t_{2},a)$, $(t_{2},b)$, $(t_{3},c)$, $(t_{3},d)$, $(t_{4},c)$,
$(a,d_1)$, $(b,d_2)$, $(b,d_3)$, $(c,d_2)$, $(d,d_2)$, $(d,d_3)\}$, and all
the weights are equal to $1$, i.e.,
$m_{p d_i}=1$ for all $p\in P$ and $1\leq i\leq 3$ such that $(p,d_i)\in E$. 

One can check that 
$V^k=(1,\phi_{k+1},\phi_{k})$, where $\phi_0=\phi_1=1$ and $\phi_{k+2}=\phi_k+\phi_{k+1}$
is the Fibonacci sequence. 
As an application of Theorem~\ref{th-1} below, it can be checked
that the value vector of this entropy game is
$V^{\infty} = (1, \varphi,\varphi)$ where $\varphi:=(1+\sqrt{5})/2$
is the golden mean.
\end{example}

\section{Stochastic mean payoff games with Kullback-Leibler payments}
\label{sec-equivalence}

We next show that entropy games are equivalent
to a class of stochastic mean payoff games in which some action spaces
are simplices, and payments are given by a Kullback-Leibler divergence.

To the entropy game $\gammage$,
we associate a stochastic zero-sum 
game with Kullback-Leibler payments, denoted $\gammampg$ and
defined as follows, referred to as ``Kullback-Leibler game''
for brevity. This new game is played
by the same players, Despot, and Tribune,
on the same weighted directed graph $G$ (so with same 
sets $E,P,D$ and same weight function $m$).
The nondeterministic half-player, People, will be replaced
by a standard probabilistic half-player, Nature.

For any node
$p\in P$, recalling that $E_p:= \{(p,d)\in E\}$ is the set
of actions available to People in state $p$, we denote by
$\Delta_p$ the set of probability measures on $E_p$.
Therefore, an element of $\Delta_p$ can be identified
to a vector $\vartheta=(\vartheta_{pd})_{(p,d)\in E_p}$ with nonnegative
entries and sum $1$. The admissible actions of Despot and Tribune 
in the states $d\in D$ and $t\in T$ are the same in the game
$\gammampg$ and in the entropy game $\gammage$. However,
the two games have different rules when the state $p\in P$
belongs to the set of People's states.
Then, Tribune is allowed to play again, by selecting
a probability measure $\vartheta\in \Delta_p$; in other words,
Tribune plays twice in a row, selecting first an arc $(t,p)\in E$,
and then a measure $\vartheta \in \Delta_p$. 
Then, Nature chooses the next state $d$
according to probability $\vartheta_{pd}$, and Tribune receives
the payment $-S_p(\vartheta;m)$, where $S_p(\vartheta;m)$ is the 
relative entropy or Kullback-Leibler divergence between
$\vartheta$ and the measure obtained by 
restricting the weight function $m$ to $E_p$:
\begin{equation}\label{KLdiv}
S_p(\vartheta; m) := \sum_{(p,d)\in E_p} \vartheta_{pd}\log (\vartheta_{pd}/m_{pd})
\enspace .
\end{equation}
Therefore, using the notations of Section~\ref{sec-genent},
we get that %
\[S_p(\vartheta; m)
= -\log\disfac{p}+\sum_{(p,d)\in E_p} \vartheta_{pd}\log (\vartheta_{pd}/q_{pd}) \]
 is minimal when the chosen probability distribution $\vartheta$ on $E_p$ is equal to the probability distribution $q_p$ of the transitions from state $p$ in the stochastic game defined in Section~\ref{sec-genent}. %
Recall that relative entropy is related to information theory and statistics~\cite{kullback}.
An interesting special case arises when $m\equiv 1$, 
as in~\cite{asarin}, thus $q_p$ is the uniform distribution on $E_p$. Then, 
$S_p(\vartheta;m)=S_p(\vartheta):= \sum_{(p,d)\in E_p} \vartheta_{pd} \log \vartheta_{pd}$
is nothing but the Shannon entropy of $\vartheta$.

A history in the game $\gammampg$ now consists of a
finite sequence $(d_0,t_0,p_0,\vartheta_0, d_1$, $t_1$, $p_1,\dots)$,
which encodes both the states and actions which have been chosen.
A strategy $\delta$ of Despot is still a function which associates
to a history ending in a state in $d$ an arc $(d,t)$  in $E$. 
A strategy of Tribune has now two components $(\tau,\pi)$, 
$\tau$ is a map which assigns to a history ending in a state in $t$ an arc
$(t,p)\in E$, as before, whereas $\pi$ assigns to the same history
and to the next state $p$ chosen according to $\tau$
a probability measure on $\Delta_p$. 

To each history corresponds a path in $G$, obtained by ignoring the 
occurrences of probability measures. For instance,
the path corresponding to the history $h=(d_0,t_0,p_0,\vartheta_0, d_1,t_1,p_1)$
is $(d_0,t_0,p_0, d_1,t_1,p_1)$.
Again, the number of turns of a history is defined as
the length of this path, each arc counting for $1/3$.
So the number of turns of $h$ is $1$ and $2/3$.
Choosing strategies $\delta$ and $(\tau,\pi)$ of both players
and fixing the initial state $d_0=\bar{d}$
determines a probability measure on the space of histories $h$.
We denote by 
\[ r^k_{\bar{d}}(\delta,(\tau,\pi)):= -\mathbb{E}\left(S_{p_0}(\vartheta_0;m)+\dots + S_{p_{k-1}}(\vartheta_{k-1};m)\right)
\]
the expectation
of the payment received by Tribune, in $k$ turns,
with respect to this measure, where $S_p$ is as in~\eqref{KLdiv} and 
$m$ is the weight function of the graph of the game.
We denote by $\valueik{\bar{d}}{k}$
the {\em value} of the game in horizon $k$,
with initial state $\bar{d}$, 
and we denote by $\valuek{k}=(\valueik{{d}}{k})_{d\in D}$ the {\em value vector}.
As in the case of entropy games, we shall use subscripts
and superscripts to indicate special versions
of the game, e.g., $\gammampg^k_d$ refers to the game in horizon
$k$ with initial state $d$. Note also our convention to use lowercase letters
(as in $\valueik{{d}}{k}$) to refer to the game with Kullback-Leibler payments,
whereas we used uppercase letters (as in $V^k_{{d}}$) to refer to the entropy game. 

It will be convenient to consider more special games in which the actions
of one of the players are restricted. We will call
{\em policy} of Despot a stationary positional strategy of this player, 
i.e., a map which assigns to every node $d\in D$ a node
$\delta(d)=t\in T$ such that $(d,t)\in E$.
Similarly, we will call {\em policy} of Tribune 
a map which assigns to every node $t\in T$
a node $\tau(t)=p\in P$ such that $(t,p)\in E$.
Observe, in this definition of policy, the symmetry between Despot and Tribune, while the game is asymetric: the policy $\tau$ is not enough to determine a positional
strategy of Tribune, because the probability distribution at every
state $p\in P$ is not specified by the policy $\tau$.
The set of policies of Despot and Tribune are denoted
by $\cP_D$ and $\cP_T$, respectively.

If one fixes
a policy $\delta$ of Despot, we end up
with a reduced game $\gammampg^k(\delta,\ast)$ in which only Tribune has actions. 
We denote by $\valuek{k}(\delta,\ast)=(\valuek{k}_d(\delta,\ast))_{d\in D}\in \R^D$ the value vector of this game in horizon $k$. Similarly, if
one fixes a policy $\tau$ of Tribune, we obtain
a reduced game denoted by $\gammampg^k(\ast,(\tau,\ast))$,  
in which Despot plays when the state is in $D$, Tribune
selects an action according to the policy $\tau$ when the state
is in $T$, and Tribune plays when the state is in $P$.
The value vector of this reduced game is denoted by
$\valuek{k}(\ast,(\tau,\ast))=(\valuek{k}_d(\ast,(\tau,\ast)))_{d\in D}\in \R^D$.
We also denote by $\valuek{k}(\delta,(\tau,\ast))=(\valuek{k}_d(\delta,(\tau,\ast)))_{d\in D}\in \R^D$ the value of the reduced game in which both policies
$\delta$ of Despot and $\tau$ of Tribune are fixed, which means that 
only Tribune plays when the state is in $P$.
The systematic character of notation used here should be self explanatory: the symbol $\ast$ refers to the actions which are not fixed by the policy.

We also consider the {\em infinite horizon}
or {\em mean payoff} game $\gammampginfty$, in which
the payment of Tribune is now
\[ r^\infty _{\bar{d}}(\delta,(\tau,\pi))
:= \limsup_{k\to \infty} k^{-1} r^k_{\bar{d}}(\delta,(\tau,\pi))\enspace .
\]
For $0<\alpha<1$, we also consider the {\em discounted game} $\gammampgd$
with a discount factor $\alpha$, in which the payment
of Tribune is 
\[ {}^{\alpha}r _{\bar{d}}(\delta,(\tau,\pi))
:= 
 -\mathbb{E}\left(S_{p_0}(\vartheta_0;m)+
\alpha  S_{p_1}(\vartheta_1;m)+ 
\alpha^2  S_{p_2}(\vartheta_2;m)+ \cdots \right)
\]
The value of the mean payoff game is denoted by 
$v^\infty_{\bar{d}}$, whereas the value of the discounted
game is denoted by ${}^{\alpha}v_{\bar{d}}$. 
As above, we denote by $\gammampginfty(\delta,\ast)$
and $\gammampginfty(\ast,(\tau,\ast))$ the games restricted
by the choice of policies $\delta,\tau$, 
and use an analogous notation for the corresponding
value vectors. For instance, $^{\alpha}v(\ast,(\tau,\ast))$ refers
to the value vector of the game ${}^\alpha\gammampg(\ast,(\tau,\ast))$
with a discount factor $\alpha$.
We define the notion of value, as well as the notion of optimal
strategies, by saddle point conditions, as in Section~\ref{sec-genent}.

 The following dynamic programming principle entails that
the value of the stochastic game with Kullback-Leibler payments
in horizon $k$ is the log of the value of the entropy game.
\begin{proposition}
\label{prop-dp}
The value vector $v^k=(v^k_d)_{d\in D}$ of the Kullback-Leibler game
in horizon $k$ does exist. 
It is determined
by the relations
$v^0=0$,  $v^k=f(v^{k-1})$, $k=1,2,\ldots$,
where 
\begin{align}\label{e-shapley}
f_d(x) = \min_{(d,t)\in E} \max_{(t,p)\in E} 
\log\left(\sum_{(p,d') \in E} m_{pd'}\exp(x_{d'})\right) \enspace ,
\end{align}
and we have $v_{d}^k = \log V_d^k$. 
\end{proposition}
In order to prove~\Cref{prop-dp}, we recall the following
classical result in convex
analysis 
showing
that the ``log-exp''function is the 
Legendre-Fenchel
transform of Shannon entropy. 
\begin{lemma}\label{lemma-rw}
The function $x\mapsto \log(\sum_{1\leq i\leq n} e^{x_i})$
is convex and it satisfies
\[
\log\left(\sum_{1\leq i\leq n} e^{x_i}\right) = \max \sum_{1\leq i\leq n} \vartheta_i (x_i - \log \vartheta_i); \qquad \vartheta_i \geq 0 , \; 1\leq i\leq n, \;
\sum_{1\leq i\leq n} \vartheta_j =1 \enspace .
\]
\end{lemma}
This result is mentioned in~\cite{RockafellarWets}, Example 11.12.
This convexity property is a
special instance of the general fact that the log of the Laplace
transform of a positive measure is convex (which follows
from the Cauchy-Schwarz inequality), whereas the explicit expression
as a maximum follows from a straightforward computation (apply Lagrange multipliers rule).

\begin{proof}[Proof of \Cref{prop-dp}]
For a zero-sum game with finite horizon and additive criterion, the existence
of the value is a standard fact, proved in a way similar
to \Cref{prop-simple}. The value vector
$v^k$ satisfies the following dynamic programming equation
\begin{align}
v^k_d &= 
\min_{(d,t)\in E} \max_{(t,p)\in E} 
\max_{\vartheta\in \Delta_p} 
\left(-S_p(\vartheta;m)
+ \langle \vartheta, v^{k-1}\rangle
\right) \enspace ,
\label{e-dp-kl}
\end{align}
where $\langle \vartheta,x\rangle = \sum_{(p,d')\in E_p} \vartheta_{pd'}x_{d'}$
for $x\in \R^D$, and $v^0_d=0$. 
By~\Cref{lemma-rw}, 
\begin{align*}
\log\left(\sum_{(p,d') \in E} m_{pd'}\exp(x_{d'})\right) 
&= \log\left(\sum_{(p,d') \in E_p} \exp(x_{d'}+\log m_{pd'})\right) \\
& = \max_{\vartheta\in \Delta_p} \sum_{(p,d')\in E_p} \vartheta_{pd'}(x_{d'}+\log m_{pd'}- \log\vartheta_{pd'})\\
&= \max_{\vartheta\in \Delta_p} 
\left(-S_p(\vartheta;m)
+ \langle \vartheta, x\rangle
\right)
\end{align*}
and so, ~\eqref{e-dp-kl} can be rewritten as $v^k=f(v^{k-1})$ where $f$ is given by~\eqref{e-shapley}.  Observe that the operator $f$ is the conjugate of the 
operator $F$ of the original entropy game: $f=\log \circ F \circ \exp$. It
follows that $v^k=f^k(v^0)=\log F^k(V^0)=\log V^k$, where for a vector $Y \in (\R^{*}_{+})^{D}$ the notation $`\log(Y)'$ denotes the vector $(\log(Y_{i}))_{1 \leq i \leq D}$, and $\exp:=\log^{-1}$. 
\end{proof}

The map $f$ arising in~\eqref{e-shapley} is obviously
order preserving and it commutes with the addition of a constant, meaning that
$f(x+\lambda \unit) = f(x) + \lambda \unit$ where $\unit$ is the unit vector $(1,...,1)^{\top}$ of $\R^{D}$,
and $\lambda \in \R$. Any map
with these two properties is nonexpansive in the sup-norm,
meaning that $\|f(x)-f(y)\|_\infty\leq \|x-y\|_\infty$,
see~\cite{crandall}. Hence, the map $x\mapsto f(x\alpha)$ has
a unique fixed point. For discounted games, the existence of the value
and of optimal positional strategies is a known fact:
\begin{proposition}\label{prop-discounted}
The discounted game $\gammampgd$ with discount factor $0<\alpha<1$ has a value and it admits
optimal strategies that are positional and stationary. The value
vector $\valued$ is the unique solution of $\valued=f(\valued\alpha)$. 
\end{proposition}
\begin{proof}
The existence and the characterization of the value are standard results, see e.g.\ the discussion in~\cite{neymansurv}. It is also known that the optimal strategies are obtained by selecting actions of the player attaining the minimum and maximum when evaluating every coordinate of $f(\valued\alpha)$, in a way similar
to the proof of \Cref{prop-simple}, $V^{k-1}$ there being replaced
by ${}^\alpha v$. Since  ${}^\alpha v$ does not depend on the number
of turns, the optimal strategies are also stationary.
\end{proof}

Nonexpansive maps can be considered more generally with respect
to an arbitrary norm. In this setting, the issue of the existence of the limit of $v^k/k= f^k(v^0)/k$ as $k\to\infty$,  and of the limit of $(1-\alpha)(\valued)$, as $\alpha\to 1^-$, where $\valued$ is the unique fixed
point of $x\mapsto f(x \alpha)$, has received much attention. The former limit is sometimes
called {\em escape rate} vector. Nonexpansiveness implies that the set
of accumulation points of the sequence $v^k/k$ is independent
of the choice of $v^0$, but it does not suffice
to establish the existence of the limit;
some additional ``tameness'' condition on the map
$f$ is needed.
Indeed, a result of Neyman~\cite{neymansurv}, using
a technique of Bewley and Kohlberg~\cite{BewlKohl76}, 
shows that the two limits
$\lim_{k\to\infty} f^k(v^0)/k$ and $\lim_{\alpha\to 1^-} (1-\alpha)\valued$
do exist and coincide if $f$ is semi-algebraic.
More generally,
Bolte, Gaubert and Vigeral~\cite{bolte2013} showed that the same
limits still exist and coincide if the nonexpansive mapping $f$ is definable in an o-minimal structure.  A counter example of Vigeral shows that the
latter limit may not exist, even if the action spaces
are compact and the payment and transition probability functions are continuous, so the o-minimality assumption is essential in what follows~\cite{Vipreprint}.

In order to apply this result, let us recall the needed definitions,
referring to~\cite{Dries98,vandendriessurvey} for background.
An {\em o-minimal structure}
consists, for each integer $n$, of a family
of subsets of $\R^n$. A subset of $\R^n$ is said to be 
{\em definable} with respect to this structure if it belongs
to this family. It is required that definable
sets are closed under the Boolean operations,
under every projection map (elimination of one variable) from $\R^n$
to $\R^{n-1}$, and under the lift,
meaning if $A\subset \R^n$ is definable, then $A\times \R\subset\R^{n+1}$ 
and $\R\times A\subset\R^{n+1}$ are also definable.
It is finally required that when $n=1$, definable subsets are precisely finite unions of intervals. A function $f$ from $\R^n$ to $\R^{k}$ is said to
be {\em definable} if its graph is definable.

An important example of o-minimal structure
is the {\em real exponential field} $\mathbb{R}_{\text{alg},\text{exp}}$.
The definable sets in this structure
are the {\em subexponential sets}~\cite{vandendriessurvey}, i.e., the images
under the projection maps $\R^{n+k}\to\R^n$ of
the {\em exponential sets} of $\R^{n+k}$, the latter being
sets of the form $\{x\mid P(x_1,\dots,x_{n+k},e^{x_1},\dots,e^{x_{n+k}})=0\}$
where $P$ is a real polynomial. A theorem of Wilkie~\cite{wilkie}
implies that $\mathbb{R}_{\text{alg},\text{exp}}$ is o-minimal,
see~\cite{vandendriessurvey}. Observe in particular that the set 
$\{x\in \R^2\mid x_1\leq x_2\}$ is definable
in this structure, being the projection of $\{x\in \R^3\mid x_2-x_1=x_3^2\}$. 
Using the o-minimal character of this structure, this 
implies that definable maps are stable by the operations of pointwise
maximum and minimum. We deduce the following 
key fact.
\begin{fact}\label{new-fact}
The dynamic programming operator $f$ of the Kullback-Leibler game, defined
by~\eqref{e-shapley}, is definable in the real exponential field.
\hfill\qed
\end{fact}

\begin{theorem}[\cite{bolte2013}]\label{th-bolte}
Let $f:\R^n \to \R^n$ be nonexpansive in any norm, and suppose
that $f$ is definable in an o-minimal structure. Then, 
\[ 
\lim_{k\to \infty} f^k(0)/k 
\]
does exists, and it coincides with 
\[
\lim_{\alpha \to 1^-} (1-\alpha) {}^\alpha v
\enspace .
\]
\end{theorem}

\begin{corollary}\label{cor-limexists}
Let $v^k=(v^k_d)_{d\in D}$ be the value vector in horizon $k$ of the stochastic
game with Kullback-Leibler payments, $\gammampg^k$, and for $0<\alpha<1$,
let $\valued$ denote the value vector of the discounted game $\gammampgd$
with discount factor $0<\alpha<1$.
Then $\lim_{k\to\infty} v^k/k$ does exist and it coincides
with $\lim_{\alpha \to 1^-} (1-\alpha)(\valued)$. 
\end{corollary}
\begin{proof}
We already noted that the map $f$ in~\eqref{e-shapley}
is nonexpansive in the sup-norm. It is definable in the
real exponential field. So Theorem~\ref{th-bolte} can be applied
to it. 
\end{proof}

\Cref{cor-limexists} will allow us to establish the existence
of the value of the mean payoff game, and to obtain
optimal strategies, by considering the discounted game,
for which, as noted in \Cref{prop-discounted}, the existence
of the value and of optimal policies are already known.

Let us recall that a strategy in a discounted game is said to be {\em Blackwell optimal} if it is optimal for all discount factors sufficiently close to one. The existence of Blackwell optimal positional strategies is a basic feature of perfect information zero-sum stochastic games with finite action spaces (see~\cite[Chap.~10]{putermanbook} for the one-player case, the two-player case builds on similar ideas, 
e.g.~\cite[Lemma~26]{gg98a}).
We next show that this result has an analogue
for entropy games. To get a Blackwell type optimality result, we need
to restrict to a setting with finitely many positional strategies.
Recall that $\cP_D$ (resp.\ $\cP_T$) denotes the set of policies
of Despot (resp.\ Tribune).
We also recall our notation
$v^\infty(\delta,\ast)$ for the value of the mean payoff game $\gammampg^\infty(\delta,\ast)$ in which
Despot plays according to the policy $\delta$.

We define the {\em projection} of a pair of strategies $(\delta,(\tau,\pi))$
in the game $\gammampg$ to be the strategy $(\delta,\tau)$ in the game $\gammage$. 
In the present setting, it is appropriate to say that a pair of 
policies $(\delta,\tau)\in \cP_D\times \cP_T$
is {\em Blackwell optimal} if there is a real number $0<\alpha_0<1$ such that,
for all $\alpha\in (\alpha_0,1)$, $(\delta,\tau)$ is the projection
of a pair of optimal strategies $(\delta,(\tau,\pi))$ in the discounted 
game ${}^\alpha\gammampg$.
\begin{theorem}\label{cor-ominimal}
The family of discounted Kullback-Leibler games 
$(\gammampgd)_{\alpha\in (0,1)}$ has positional Blackwell optimal strategies.
\end{theorem}
\begin{proof}
For all $\alpha\in (0,1)$, the {\em discounted} game has positional optimal strategies
$\delta^*,(\tau^*,\pi^*)$.
This follows from the standard dynamic programming
argument mentioned in the proofs of Proposition~\ref{prop-simple} and
\ref{prop-discounted}, noting that $\delta^*(d)$ is obtained 
by choosing any $t\in T$ such that $(d,t)\in E$ attains the minimum 
in the expression
\[
{}^{\alpha}v_d = 
\min_{(d,t)\in E} \max_{(t,p)\in E} 
\max_{\vartheta\in \Delta_p} 
\left(-S_p(\vartheta;m)
+ \langle \vartheta, \alpha ({}^\alpha v)\rangle
\right) \enspace .
\]
Similarly, $\tau^*(t)$ is chosen to be any $p\in P$ such that $(t,p)\in E$
attains the maximum in
\[ \max_{(t,p)\in E} \max_{\vartheta\in \Delta_p} 
\left(-S_p(\vartheta;m)
+ \langle \vartheta, \alpha ({}^\alpha v)\rangle
\right) \enspace ,
\]
and $\pi^*(p)$ is chosen to be the unique action $\vartheta$
attaining the maximum in
\[
\max_{\vartheta\in \Delta_p} 
\left(-S_p(\vartheta;m)
+ \langle \vartheta, \alpha ({}^\alpha v)\rangle
\right) \enspace 
\]
(observe that the function to be maximized is strictly concave and continuous
on $\Delta_p$,
and that $\Delta_p$ is compact and convex, so the maximum
is achieved at a unique point). 

By definition of the value and of optimal
strategies, we have, for all strategies $\delta$ and $(\tau,\pi)$
of Despot and Tribune respectively, 
\begin{align} {}^{\alpha}r _{d}(\delta^*,(\tau,\pi))
\leq \valued_d= {}^{\alpha}r _{d}(\delta^*,(\tau^*,\pi^*))
\leq {}^{\alpha}r _{d}(\delta,(\tau^*,\pi^*))\enspace , \label{e-saddle}
\end{align}
which is equivalent to 
\begin{align}
\valued_d= {}^{\alpha}r _{d}(\delta^*,(\tau^*,\pi^*))=
\valued_d(\delta^*,\ast)=\valued_d(\ast,(\tau^*,\pi^*))\enspace .
\label{e-saddle2}\end{align}
Specializing the first inequality in~\eqref{e-saddle}
 to $\tau=\tau^*$, and bounding above the last term, we deduce
that, for all for all strategies $\delta$ and $\tau$
of Despot and Tribune respectively, we have
\begin{equation}
\valued(\delta^*,(\tau,\ast))
\leq \valued_d=\valued(\delta^*,(\tau^*,\ast))
\leq \valued(\delta,(\tau^*,\ast))\enspace ,
\label{e-saddle-pol}
\end{equation}
where $\valued_d(\delta,(\tau,\ast))$ is the 
value of the reduced discounted 1-player  discounted game
 $^\alpha\gammampg(\delta,(\tau,\ast))$ starting at $d\in D$, 
in which the (not necessarily positional)
strategies  $\delta$ of Despot and $\tau$ of Tribune are fixed.
The inequalities~\eqref{e-saddle-pol} can be specialized
in particular to policies $\delta\in \cP_D$ and $\tau\in\cP_T$. 
Then, %
by \Cref{prop-discounted}, 
$\valued(\delta,(\tau,\ast))$ is the unique
fixed point of the self-map $x\mapsto {}^{\tau}f^{\delta}_d(x\alpha)$
of  $\R^D$, 
where ${}^\tau f^{\delta}$ is the dynamic programming operator
 given by 
\begin{align}
{}^{\tau}f^{\delta}_d(x) = 
\log\left(\sum_{(\tau\circ \delta(d),d') \in E} m_{\tau\circ \delta(d),d'}\exp(x_{d'})\right) \enspace .
\label{e-Mdeltatau}
\end{align}
It follows that the map $\alpha\mapsto \valued(\delta,(\tau,\ast))$
is definable in the real
exponential field $\mathbb{R}_{\text{alg},\text{exp}}$.
(To see this, observe that, by \Cref{new-fact},
the set $\{(x,y)\mid x={}^{\tau}f^{\delta}_d(y)\}\times \R$
is definable in this structure; then, taking the intersection
of this set with the definable sets $\{(x,y,\alpha)\mid y_d = x_d \alpha\}$,
for $d \in D$, and projecting the intersection
keeping only the $x$ and $\alpha$ variables, we obtain a definable
set which is precisely the graph of the map $\mapsto \valued(\delta,(\tau,\ast))$).

For all $(\bar{\delta},\bar{\tau})\in\cP_D\times \cP_T$, let
$I(\bar{\delta},\bar{\tau})$ denote the set of $\alpha\in (0,1)$ such that
\begin{equation}
\valued(\bar{\delta},(\tau,\ast))
\leq \valued(\bar{\delta},(\bar{\tau},\ast))
\leq \valued(\delta,(\bar{\tau},\ast))
\label{blackwelleq}
\end{equation}
holds for all $(\delta,\tau)\in\cP_D\times\cP_T$. 
Since the saddle point property~\eqref{e-saddle-pol}
holds for all $\alpha$ ($\delta^*$ and $\tau^*$ depend
on $\alpha$, of course), we have
\begin{align}
\cup_{(\bar{\delta},\bar{\tau})} I(\bar{\delta},\bar{\tau})= (0,1)\enspace .
\label{e-covering}
\end{align}
Observe that the set $I(\bar{\delta},\bar{\tau})$ is a subset of $\R$
definable in the real exponential field,
which is o-minimal.
It follows that $I(\bar{\delta},\bar{\tau})$ is a finite union of intervals.  Hence, ~\eqref{e-covering} provides
a covering of $(0,1)$ by finitely many intervals, and so,
one of the sets $I(\bar{\delta},\bar{\tau})$ must include
an interval of the form $(1-\epsilon,1)$. 

To show that the policies $\bar{\delta},\bar{\tau}$
obtained in this way are Blackwell optimal,
it remains to show that if $(\bar{\delta},\bar{\tau})$ 
satisfies~\eqref{blackwelleq} for some $\alpha$,
 then it is the projection
of a pair of optimal strategies $(\bar{\delta},(\bar{\tau},\bar{\pi}))$
in the discounted game ${}^\alpha\gammampg$.
For this, we shall apply the existence of optimal
strategies that are positional and the resulting equations~\eqref{e-saddle2} 
and~\eqref{e-saddle-pol}
to the reduced  games
 $^\alpha\gammampg(\bar{\delta},\ast)$ and
 $^\alpha\gammampg(\ast, (\bar{\tau},\ast))$, respectively.

The first game leads to the existence of positional
stationary strategies $\tau^1,\pi^1$ of Tribune
such that, for all $d\in D$, 
\[\valued_d(\bar{\delta},\ast)= {}^{\alpha}r _{d}(\bar{\delta},(\tau^1,\pi^1))
=\valued(\bar{\delta},(\tau^1,\ast))\enspace .\] 
Then, using~\eqref{blackwelleq}, we get that
$\valued(\bar{\delta},(\bar{\tau},\ast))\geq  
\valued(\bar{\delta},(\tau^1,\ast))=\valued_d(\bar{\delta},\ast)\geq
\valued(\bar{\delta},(\bar{\tau},\ast))$, 
hence the equality
$\valued(\bar{\delta},(\bar{\tau},\ast))=\valued_d(\bar{\delta},\ast)$.

The second one leads to the existence of positional
stationary strategies $\delta^2,\pi^2$ of Despot and Tribune respectively
such that, for all $d\in D$, 
\[ %
\valued_d(\ast, (\bar{\tau},\ast))=
 {}^{\alpha}r _{d}(\delta^2,(\bar{\tau},\pi^2))=
\valued_d(\delta^2,(\bar{\tau},\ast))=\valued_d(\ast,(\bar{\tau},\pi^2))
\enspace .\] %
Then, using~\eqref{blackwelleq}, we deduce that
$\valued_d(\bar{\delta},(\bar{\tau},\ast))\leq
\valued_d(\delta^2,(\bar{\tau},\ast))=\valued_d(\ast, (\bar{\tau},\pi^2)) 
\leq \valued_d(\bar{\delta},(\bar{\tau},\pi^2))
\leq \valued_d(\bar{\delta},(\bar{\tau},\ast))$, hence the equality
$\valued_d(\bar{\delta},(\bar{\tau},\ast))
=\valued_d(\ast, (\bar{\tau},\pi^2)) 
=\valued_d(\bar{\delta},(\bar{\tau},\pi^2))$.
With the equality proved with the first game, this leads to
$\valued_d(\bar{\delta},\ast)=\valued(\bar{\delta},(\bar{\tau},\ast))  
=\valued_d(\bar{\delta},(\bar{\tau},\pi^2)) 
=\valued_d(\ast, (\bar{\tau},\pi^2))$.
This shows that $(\bar{\delta},(\bar{\tau},\pi^2))$ is  a pair of 
 optimal strategies for the discounted game ${}^\alpha\gammampg$.
Since $(\bar{\delta},\bar{\tau})$ is its projection, we get that
it is Blackwell optimal.
\end{proof}

\begin{theorem}\label{cor-morphism}
The value $v^\infty$ of the stochastic
mean payoff game with Kullback-Leibler payments does exist,
and it coincides with $\lim_{k\to\infty} v^k/k = \lim_{\alpha \to 1^-} (1-\alpha)(\valued)$.  For all $(\delta,\tau)\in \cP_D\times \cP_T$, the same properties hold for the values $v^\infty(\delta,\ast)$  and $v^\infty(\ast,(\tau,\ast))$ of the reduced games in which Despot plays according to $\delta$ when the state is in $D$ and Tribunes plays according to $\tau$ when the state is in $T$, respectively.
Moreover, 
the Blackwell optimal strategies  $(\delta^*,\tau^*)\in \cP_D\times \cP_T$
of~\Cref{cor-ominimal} satisfy, 
for all $d\in D$,
\begin{align}
v^\infty_d=v^\infty_d(\delta^*,(\tau^*,\ast)) = v^\infty_d(\delta^*,\ast) = v^\infty_d(\ast,(\tau^*,\ast)) \enspace .\label{e-forv1}\end{align}
In particular,
\begin{align*}
v^\infty_d
&= \min_{\delta \in \cP_D} v^\infty_d(\delta,\ast) = \max_{\tau\in \cP_T}
v^\infty_d(\ast,(\tau,\ast)) \enspace.
\end{align*}
\end{theorem}
\begin{proof}%
We already noted in \Cref{cor-limexists} that
\begin{align}
\lim_{k\to\infty} v^k/k = \lim_{\alpha \to 1^-} (1-\alpha)(\valued)\enspace .
\label{e-e0}
\end{align}
Moreover, it is shown in \cite[Corollary 2, (iii)]{bolte2013},
as a consequence of a theorem of Mertens and Neyman~\cite{MertNeym81}, that this
limit coincides with the value of the game. These results
rely on the definable and sup-norm nonexpansive 
character of the dynamic programming
operator $f$. The dynamic programming operator $f^\delta$, 
associated to the reduced game determined by the strategy $\delta\in\cP_D$
can be written as 
\begin{align}
f^\delta_d(x):= \max_{(\delta(d),p)\in E}
\log\left(\sum_{(p,d') \in E} m_{pd'}\exp(x_{d'})\right) \enspace .
\label{e-fdelta}
\end{align}
It is definable and sup-norm nonexpansive, hence the same conclusions
apply to the game $\gammampg(\delta,\ast)$, i.e,
\begin{align}
\lim_{k\to\infty} v^k(\delta,\ast)/k = \lim_{\alpha \to 1^-} (1-\alpha)(\valued(\delta,\ast))
\label{e-14}
\end{align}
is the value of the game $\gammampg^\infty(\delta,\ast)$.
We argue in the same way for the game $\gammampg(\ast,(\tau,\ast))$, noting
that the associated dynamic programming operator is now
\begin{align}
{}^{\tau}f_d(x):= \min_{(d,t)\in E} 
\log\left(\sum_{(\tau(t),d') \in E} m_{\tau(t)d'}\exp(x_{d'})\right) \enspace .
\label{e-e1}
\end{align}
which is still definable and sup-norm nonexpansive.
Hence, 
\begin{align}
\lim_{k\to\infty} v^k(\ast,(\tau,\ast))/k = \lim_{\alpha \to 1^-} (1-\alpha)(\valued(\ast,(\tau,\ast)))\label{e-e2}
\end{align}
is the value of the game $\gammampg^\infty(\ast,(\tau,\ast))$.

Let $\delta^*,\tau^*$ denote the positional Blackwell optimal strategies
constructed in \Cref{cor-ominimal}. 
By definition,
there is an interval $(\alpha_0,1)$ such that for all $\alpha \in(\alpha_0,1)$,
there exists $\pi^*$ depending on $\alpha$ such that
\begin{align*}
\valued=\valued(\delta^*,(\tau^*,\pi^*))=\valued(\delta^*,\ast)
= \valued(\ast,(\tau^*,\pi^*)) 
\end{align*}
which by~\eqref{e-saddle-pol} leads to
\begin{align}
\valued=\valued(\delta^*,\ast)
= \valued(\ast,(\tau^*,\ast)) =\valued(\delta^*,(\tau^*,\ast))\enspace .
\label{e-forz1}
\end{align}
Multiplying these expressions by $(1-\alpha)$, passing
to the limit, and using~\eqref{e-e0},~\eqref{e-14} and
\eqref{e-e2}, 
we obtain~\eqref{e-forv1}.
\end{proof}
\begin{remark}
\Cref{cor-morphism} shows that Player Despot has an optimal
positional strategy $\delta^*$ in the mean payoff Kullback-Leibler game.
It also shows that in the same game, the actions of Player Tribune at states
$t\in T$ can be chosen according to the optimal positional strategy $\tau^*$.
This theorem does not imply, however, that at every state $p\in P$, 
the optimal action $\vartheta \in \Delta_p$ can be chosen
optimally according to a positional strategy $\pi$.
Indeed, the proof
by an o-minimality argument uses in an essential way the fact that
there are finite actions spaces at every states $d$ and $t$, whereas
$\Delta_p$ is infinite. We leave for further investigation
the question of the existence of such a positional strategy,
noting that it is not needed in the application 
to entropy games. 
\end{remark}

\begin{remark}
It is shown in  \cite[Corollary 2, (iii)]{bolte2013}, as a consequence
of a theorem of Mertens and Neyman~\cite{MertNeym81},
that a stochastic game
with a definable Shapley operator has a {\em uniform value}, a property
which is stronger than the mere existence of the value.
Loosely speaking, a stochastic game 
with initial state $d$ is said to have a uniform value $v^{\infty}_d$ if both players can almost guarantee  $v^{\infty}_d$ provided that the length of the $k$-stage game is large enough. In the present setting, we get
the following property: for any $\epsilon>0$, there is a couple of strategies
of $(\delta, (\tau,\pi))$ and a time $K$ such that, for every $k\geq K$, every starting state $d$ and every strategies $\delta '$ and $(\tau',\pi')$,
\begin{eqnarray*}
r^k_d(\delta,\tau',\pi')/k\leq v^\infty_d+\epsilon, \qquad 
r^k_d(\delta',\tau,\pi)/k \geq v^\infty_d-\epsilon \enspace .
\end{eqnarray*}
\end{remark}

\section{Application to the entropy game model}
\subsection{Existence of optimal positional strategies in the entropy game}
As an application of \Cref{cor-morphism}, we obtain the existence
of optimal positional strategies in the entropy game model
of \Cref{sec-genent}.

\begin{theorem}\label{th-1}
The infinite horizon entropy game has a value and it has optimal
positional strategies, namely the
Blackwell optimal strategies  $(\delta^*,\tau^*)\in \cP_D\times \cP_T$
of~\Cref{cor-ominimal}. Moreover, for all initial states $d$,
\[
V_d^\infty = \lim_{k\to\infty} (V_d^k)^{1/k}\enspace .
\]
\end{theorem}

\begin{proof}
By \Cref{prop-simple}, $V^k=F(V^{k-1})$ where $F$ is as in~\eqref{e-def-dp}.
Moreover, $F=\exp\circ f\circ \log$
is the conjugate of the dynamic programming operator $f$ of the Kullback-Leibler
game introduced in \eqref{e-shapley}. \Cref{cor-limexists} shows
that $v^\infty=\lim_{k\to\infty} v^k/k$ does exists. It follows that
$V^\infty_d:=\lim_{k\to\infty} (V_d^k)^{1/k} = \exp(v^\infty_d)$ does
exist for all $d\in D$. 

Let $(\delta^*,\tau^*)$ denote the Blackwell optimal strategies given by
\Cref{cor-ominimal}. We showed in \Cref{cor-morphism}
that $v^\infty= v^\infty(\delta^*,\ast)$, and by \Cref{cor-limexists}, we have
\[
 v^\infty(\delta^*,\ast)= \lim_{k\to\infty} v^k(\delta^*,\ast)/k\enspace .
\]
Using the dynamic programming principle for finite horizon 1-player games,
or equivalently, by applying \Cref{prop-dp} to the reduced finite horizon 
game $\gammampg^k(\delta^*,\ast)$, we obtain that 
\[  v^k(\delta^*,\ast)/k= (f^{\delta^*})^k(0)/k \enspace ,\]
where for all $\delta$, $f^\delta$ is the dynamic programming operator defined in~\eqref{e-fdelta} associated to the reduced game $\gammampg(\delta,\ast)$. Consider now
the conjugate $F^\delta:=\exp\circ f^\delta \circ \log$, so that
\begin{equation}\label{defFdelta}
F^\delta_d(X):= \max_{(\delta(d),p)\in E}
\left(\sum_{(p,d') \in E} m_{pd'}X_{d'}\right) \enspace .
\end{equation}
By Proposition~\ref{prop-simple}, 
for all strategies $\tau$ of Tribune, non necessarily positional,
and all initial states $d\in D$,
we have
\[
R^k_d(\delta^*,\tau) \leq [(F^{\delta^*})^k(e)]_d \enspace .
\]
Applying all the above equalities and inequalities, we deduce that
\[
\limsup_{k\to\infty} (R^k_d(\delta^*,\tau))^{1/k} \leq 
\lim_{k\to\infty} [(F^{\delta^*})^k(e)]_d^{1/k}=\exp(v^\infty_d)=
 V_d^\infty \enspace,
\] 
so Player Despot can guarantee his loss does not exceed
$V_d^\infty$ by playing $\delta^*$. 

Let us now consider the reduced infinite horizon or finite 
horizon entropy and Kullback-Leibler games
in which the strategy of Tribune is fixed and equal to $\tau^*$.
By the same arguments as above, we  show
 that the positional strategy $\tau^*$ guarantees
to Player Tribune to win at least $V_d^\infty$ in the entropy game.
Indeed, applying successively 
\Cref{cor-morphism} and \Cref{prop-dp}, we deduce 
\[
v^\infty= v^\infty(\ast,(\tau^*,\ast))=\lim_{k\to\infty} v^k(\ast,(\tau^*,\ast))/k
= \lim_{k\to\infty} ({}^{\tau^*}f)^k(0)/k\enspace , \]
where, for all $\tau$, ${}^{\tau}f$ is as in~\eqref{e-e1}.
Then, considering 
the conjugate ${}^{\tau}F:=\exp\circ {}^{\tau}f \circ \log$, so that,
for all $\tau$,
\begin{equation}\label{defFtau}
{}^{\tau}F_d(X):= \min_{(d,t)\in E}
\left(\sum_{(\tau(t),d') \in E} m_{\tau(t)d'}X_{d'}\right) \enspace ,
\end{equation}
and applying Proposition~\ref{prop-dp} and~\ref{prop-simple},
we deduce that
\[ V^\infty_d=\lim_{k\to\infty} ({}^{\tau^*}F_d)^k(e)^{1/k}
\leq\liminf_{k\to\infty} R^k_d(\delta,\tau^*)^{1/k} \enspace ,\]
for all strategies $\delta$ of Despot, non necessarily positional,
and all initial states $d\in D$.
So Player Tribune can win at least $V_d^\infty$ by playing
$\tau^*$.
\end{proof}

\subsection{Comparison with the original entropy game model}
\label{sec-orig}
The original entropy game model of Asarin et al.~\cite{asarin}
is a zero-sum game defined in a way similar to \Cref{sec-genent},
up to a technical difference:
in their model, the initial state is not prescribed.
The payment of Tribune in horizon $k$, instead
of being $R_{\bar{d}}^k(\delta,\tau)$, is 
the quantity $\bar{R}^k(\delta,\tau)$, defined now as the sum 
of weights of all paths of length $k$ starting at a node in $D$
and ending at a node in $D$. Hence, 
\(
\bar{R}^k(\delta,\tau)= \sum_{d\in D} R^k_d(\delta,\tau)\).
The payment of Tribune can be defined in their game as follows
\(
\bar{R}^\infty(\delta,\tau) = \limsup_{k\to\infty} (\bar{R}^k(\delta,\tau))^{1/k}
\).
This game is denoted by $\gammaeinfty$, we denote by $\bar{V}^\infty$ the value of this game, which is shown to exist in~\cite{asarin}. 

Note that in the initial model in~\cite{asarin}, the weights $m_{pd'}$ are equal to $1$. The generalization to weighted entropy games,
in which the weights $m_{pd'}$ are integers
is discussed in Section 6 of~\cite{asarin}.
The case in which the weights $m_{pd'}$ take rational values
can be reduced to the latter case
by multiplying all the weights by an integer factor.
Therefore, we will ignore the restriction that $m_{pd'}=1$
in our definition of $\gammaeinfty$ and will refer to the entropy
game model with rational weights as the entropy game model.
The following observation follows readily from \Cref{th-1}.
\begin{proposition}\label{prop-compare}
The value of the original entropy game $\gammaeinfty$
considered by Asarin et al.~\cite{asarin} (with a free initial state),
coincides with the maximum
of the values of the games $\gammage^{\infty}_d$,
taken over all initial states $d\in D$:
\[
\bar{V}^\infty = \max_{d\in D} V_d^\infty 
\enspace .
\]
\end{proposition}
\begin{example}
\Cref{prop-compare} is illustrated by the game of Example~\ref{ex-1}. 
In the original model of \cite{asarin},
the value, defined independently of the initial state,
is $(1+\sqrt{5})/2$, whereas our model associates to the initial
state $d_1$ a value $1$ which differs from 
the values of $d_{2}$ and $d_{3}$.
\end{example}

In~\cite{asarin}, entropy games were compared with matrix multiplication games.
We present here this correspondence in the case of general weights $m_{pd'}$.
Given policies $\delta\in\cP_D$ and $\tau\in \cP_T$,
let $A(\delta)\in\R^{D\times T}$ and $B(\tau)\in \R^{T\times D}$ be such that 
$A(\delta)_{dt}=1$ if $t=\delta(d)$ and $0$ otherwise,
and $B(\tau)_{td}=m_{\tau(t)d}$ if $(\tau(t),d)\in E$ 
and $0$ otherwise, for all $(d,t)\in D\times T$.
We shall think of $A(\delta)$ and $B(\tau)$ as rectangular matrices.
Then 
\( \bar{R}^k(\delta,\tau)= \|(A(\delta)B(\tau))^k\|_1\),
where for any $A\in\R^{D\times D}$, $A^k$ denotes its $k$th power
and $\|A\|_1=\sum_{dd'}|A_{dd'}|$ its $\ell^1$ norm.
From this, one deduces that 
\(
\bar{R}^\infty(\delta,\tau)=\rho(A(\delta)B(\tau)),\)
where $\rho(A)$ denotes the spectral radius of the matrix $A$.
Moreover, let
$\A$ and $\B$ denote the sets of all matrices of the form $A(\delta)$ and
$B(\tau)$ respectively, and let $\A\B$ be the set of all matrices
$AB$ with $A\in\A$ and $B\in\B$.
The sets $\A$, $\B$ and $\A\B$ are subsets of matrices $\M$
satisfying the property that all elements of $\M$ have same dimension
and if $\M_i$ is the set of $i$th rows of the elements of $\M$,
then $\M$ is the set of matrices the $i$th row of which belongs to 
$\M_i$.
Such a property defines the notion of IRU matrix sets (for independent row
uncertainty sets) in~\cite{asarin}.
The following property proved in~\cite{asarin} is the analogue
of Theorem~\ref{cor-morphism}, $V_d^\infty$ being replaced by $\bar{V}^\infty$:
\begin{equation}\label{eq-IRU}
\bar{V}^\infty= 
\min_{A\in \A} \max_{B\in\B} \rho(AB)= \max_{B\in\B}\min_{A\in \A} \rho(AB)\enspace. 
\end{equation}
A more general property is proved in~\cite[Section~8]{agn},
as a consequence of the Collatz-Wielandt theorem (see Corollary~\ref{th-cw} below).

\begin{remark}
Our approach shows that entropy games reduce to Kullback-Leibler
games, which are stochastic mean payoff games (with compact action spaces), 
Asarin et al.~\cite{asarin} remarked that the special {\em deterministic}
entropy games, in which People has only one possible
action in each state, can be re-encoded as deterministic
mean payoff games. This can also be recovered
from our approach: in this deterministic case, the simplices
$\Delta_p$ are singletons in the Kullback-Leibler game,
and the entropy function vanishes, so the Kullback-Leibler
game degenerates to a deterministic mean payoff game.
\end{remark}

\section{Applying the Collatz-Wielandt theorem to entropy games}
The classical Collatz-Wielandt formul\ae\ provide variational
characterizations of the spectral radius $\rho(M)$ of a nonnegative matrix
$M\in \R^{D\times D}$:
\begin{align*}
\rho(M)&= 
\inf \{\lambda >0\mid \exists X\in \operatorname{int}\R_+^D, \; MX \leq \lambda X \} \\
&=  \max\{\lambda \geq 0\mid \exists X\in \R_+^D\setminus\{0\}, \; MX = \lambda X \}\enspace,\\
&=  \max\{\lambda \geq 0\mid \exists X\in \R_+^D\setminus\{0\}, \; MX \geq \lambda X \}\enspace,
\end{align*}
where $\R_+^D$ denotes the nonnegative orthant of $\R^D$, and $\operatorname{int}\R_+^D$ its interior, i.e, the set of positive vectors. The infimum is not always attained in the first line, whereas by writing ``max'', we mean that the suprema are always attained.

This has been extended to non-linear, order preserving
and continuous self-maps of the standard positive 
cone by Nussbaum~\cite{nussbaum86}, see also~\cite{arxiv1,GV10,AGGut10,agn,lemmens}. 
Recall that a self-map of $\R^{D}_{+}$ is said to be {\em order preserving}
if $u \leq v \Rightarrow F(u) \leq F(v)$ for all $u,v \in \R^{D}_{+}$, the relation $\leq$ being understood entrywise. This map is {\em positively homogeneous of degree $1$} if $F(\alpha u) = \alpha u$, for all $\alpha >0$ and $u \in \R^{D}_{+}.$
\begin{theorem}\label{cw-0}
Let $F$ be a continuous, order preserving, and positively homogeneous of degree $1$ self-map of $\R_+^D$. Then the following quantities coincide
\begin{align}
&\lim_{k\to\infty} \max_{d\in D}[F^k(e)]_d^{1/k} \label{e-lim1}\\
& \inf \{\lambda >0\mid \exists X\in \operatorname{int}\R_+^D, \; F(X) \leq \lambda X \} \label{e-lim2}\\
&  \max\{\lambda \geq 0\mid \exists X\in \R_+^D\setminus\{0\}, \; F(X) = \lambda X \}\enspace,\label{e-lim3}\\
&  \max\{\lambda \geq 0\mid \exists X\in \R_+^D\setminus\{0\}, \; F(X) \geq \lambda X \}\enspace,\label{e-lim4}
\end{align}
\end{theorem}
\begin{proof}
The existence of~\eqref{e-lim1} and the fact it coincides
with~\eqref{e-lim2} is proved in~\cite[Prop~1]{arxiv1}. The fact
that~\eqref{e-lim2} and~\eqref{e-lim3} coincide is proved
in \cite[Theorem~3.1]{nussbaum86}. 
The fact that~\eqref{e-lim3}
and~\eqref{e-lim4} coincide is proved
in~\cite[Lemma~2.8]{AGGut10}.
\end{proof}
The common value of the expressions in~\Cref{cw-0} is called the non-linear spectral radius
of $F$.

We showed in \Cref{prop-compare} that the value of the original entropy game
$\gammae$ is precisely
\[
\bar{V}^\infty = \max_{d\in D} V_d^\infty 
= \lim_{k\to\infty} \max_{d\in D} [F^k(e)]_d^{1/k}
\]
where $F$ is the dynamic programming operator defined in~\eqref{e-def-dp}.
This operator is continuous, order preserving, and homogeneous
of degree one.
Hence, we get as an immediate corollary of \Cref{cw-0}:
\begin{corollary}\label{th-cw}
The value $\bar{V}^\infty$ of the original entropy game $\gammae$ (with a free initial state)
coincides with any of the following expressions
\begin{align}
&\inf\{\lambda >0\mid \exists X\in \operatorname{int}\R_+^D, \; F(X) \leq \lambda X \} \label{e-cw1}\\
&\max\{\lambda >0\mid \exists X\in \R_+^D\setminus\{0\}, \; F(X) = \lambda X \}\label{e-cw2}\\
& \max\{\lambda >0\mid \exists X\in \R_+^D\setminus\{0\}, \; F(X) \geq \lambda X \} \enspace ,\label{e-cwlast}
\end{align}
where $F$ is the dynamic programming operator~\eqref{e-def-dp}.
\end{corollary}
The Collatz-Wielandt formul\ae\ of \Cref{th-cw} are helpful to establish strong
duality results, like~\eqref{eq-IRU}.
Note that \eqref{eq-IRU} is weaker than \Cref{th-cw} since it does not
imply the existence of a nonlinear vector whereas~\eqref{e-cw2} 
of \Cref{th-cw} does.
 See also~\cite{AGGut10} for
an application to mean payoff games and tropical geometry.
Our main interest here lies in the following application
of~\eqref{e-cw1}.
We say that a state $d$ of Despot is {\em significant}
if the set of actions of Despot in this state, $\{(d,t)\in E\}$, 
has at least two elements (i.e., Despot has to make a choice
in this state).
We say that an entropy game is {\em Despot-free} if the Despot
player does not have any significant state. A Despot-free
game is essentially a one (and half) player problem, since the minimum
term in the corresponding dynamic programming operator~\eqref{e-def-dp} vanishes.
Indeed, for each $d\in D$, there is a unique node $t$
such that $(d,t)\in E$, and we define the map $\sigma: D\to T$ by $\sigma(d)=t$.
The following corollary, which follows from 
Corollary~\ref{th-cw} by making the change of variables $\mu =\log \lambda$ and $x=\log X$, is also a special case of a result of Anantharam and Borkar~\cite{anantharam}.
\begin{corollary}\label{cor-cw}
The logarithm of the value of a Despot-free original entropy game
is given by the value of the optimization problem
\begin{align}
&\inf \mu \; \nonumber\\
&(\mu,x) \in \R\times\R^D, \text{satisfying}\nonumber\\
& \mu  + x_d \geq \log(\sum_{d'\in D} m_{p,d'}e^{x_{d'}}) 
\text{ for all } d \in D, p\in P \text{ such that } (\sigma(d),p)\in E
\enspace. \label{e-cvx0}
\end{align}
\end{corollary}
Observe that the latter expression is the value of an optimization
problem in which the variables are $\mu$ and $x=(x_d)_{d\in D}$,
the objective function is the linear form $(\mu,x)\mapsto \mu$,
and the feasible set is convex.  Hence, this will lead us to a polynomial time
decision procedure in the Despot free case, which we develop
in the next section. 

\section{Polynomial time solvability of entropy games with a few significant Despot positions}\label{sec-few}

By {\em solving strategically} an entropy game, we mean
finding a pair of optimal policies. We assume
from now that the weights $m_{p,d}$ are integers.
Since policies are combinatorial objects, solving strategically the game
is a well posed
problem in the Turing (bit) model of computation. Once
optimal policies are known, 
the value of the game, which is an algebraic number,
can be obtained as the Perron root
of an associated integer matrix. 
Our first main result is the following.
\begin{theorem}\label{th-polytime}
Despot-free entropy games can be solved strategically
in polynomial time. 
\end{theorem}

This will be proved in \Cref{sec-proof}, by combining
several ingredients: a reduction to the irreducible
case, an application of the ellipsoid method, and 
separation bounds for algebraic numbers. 

We will also show  the following generalization
of \Cref{th-polytime}.
\begin{theorem}\label{cor-polytime}
Entropy games in which Despot has a fixed number of significant
states can be solved strategically in polynomial time. 
\end{theorem}

\section{Proof of Theorem~\ref{th-polytime} and Theorem~\ref{cor-polytime}}
\label{sec-proof}
We start by considering a Despot-free game. 
We decompose the proof of \Cref{th-polytime} in several steps, corresponding to different
subsections.  
\subsection{Reduction to the irreducible case}
First, we associate to a Despot-free entropy game a projected
directed graph $\bar G$, with node set $D$ and an arc $d\to d'$ if
there is a path $(d,t,p,d')$ in the original directed graph $G$.
We say that the game is {\em irreducible} if $\bar G$ is strongly
connected.  
Recall that we assumed that any node has a successor, 
so that strongly connected
components are all non trivial (not reduced to a node with no edge).

\begin{lemma}\label{e-valueunique}
The value of an irreducible Despot-free entropy game is independent
of the initial state. 
Moreover, there is a vector $U\in \operatorname{int}\R_+^D$
and a scalar $\lambda^*>0$ such that $F(U)=\lambda^* U$, and $\lambda^*$
coincides with the value of any initial state in this game.
\end{lemma}
\begin{proof}
The non-linear Perron-Frobenius theorem in~\cite{arxiv1}
provides a sufficient condition for the existence
of a positive eigenvector of an order preserving
positively homogeneous self-map $F$ of the interior
of the cone. It suffices to check that a certain directed graph $G(F)$
associated to $F$ is strongly connected. Specialized
to the present setting, this directed graph is defined as follows:
the node set is $D$, 
and there is an arc from $d\to d'$ if
\[ \lim_{s\to\infty}F_d(se_{d'})=+\infty
\enspace ,
\]
where
$e_{d'}=(0,\dots,0,1,0,\dots,0)^\top$ is the $d'$th vector
 of the canonical basis of $\R^D$. By considering the explicit
form of $F$, we see that $G(F)$ is precisely the directed graph $\bar{G}$.
Hence, by Theorem~2 of~\cite{arxiv1},
there exists a vector $U\in \operatorname{int}\R_+^D$
and a scalar $\lambda^*>0$ such that $F(U)=\lambda^* U$. 
It follows from \Cref{th-cw} and from the fact that the
value of the original entropy game is $\max_{d\in D}V^\infty_d$ that $V^\infty_d\leq \lambda^*$ holds for all $d\in D$.

We next show that the other inequality holds.
Recall that $V^\infty_d = \lim_{k\to\infty} [F^k(e)]_d^{1/k}$ is the value of the entropy game with initial state $d$, where $e=(1,...,1)^{\top} \in \R^{D}_{+}$.

Since $D$ is finite, there is a positive scalar $\beta>0$ such that
$e\geq \beta U$  (indeed, take $\beta:=(\max_{d\in D} U_d)^{-1}$). Using
the order preserving and positively homogeneous character
of the dynamic programming operator $F$, we get
\[
F^k(e) \geq F^k(\beta U) = \beta F^k(U)=\beta (\lambda^*)^k U
\]
and so, using that $U$ is positive, 
\[
V^\infty_d=\lim_{k\to\infty}[F^k(e) ]_d^{1/k}\geq \lim_{k\to\infty} 
(\beta (\lambda^*)^k U_d)^{1/k} = \lambda^* \enspace. 
\]
Hence, 
$V^\infty_d= \lambda^*$ holds for all $d\in D$.
\end{proof}
Thanks to this lemma, we will speak of ``value'',
without mentioning the initial state, when the
entropy game is irreducible.

Every strongly connected component $C$ of $\bar G$
with set of nodes $D_C\subset D$ yields a reduced game, in which
the set of states of Despot is $D_C$, and Tribune only selects actions 
such that the next state of Despot will remain in $D_C$ 
for at least one action chosen by People. Moreover, People
chooses only actions so that the next state remains in $D_C$.
By definition of $\bar G$, this reduced
game is irreducible. We denote it by $\gammage[C]$. 
The following elementary observation allows us to reduce
the general Despot-free case to the irreducible Despot-free case.
\begin{lemma}
In a Despot-free entropy game,
the value of a state $d$ is the maximum of the value of the irreducible
games $\gammage[C]$ corresponding to the different strongly connected
components $C$ of $\bar G$ to which $d$ has access.
\end{lemma}
\begin{proof}
  This follows from a more precise of Zijm, Theorem~5.1 in~\cite{zijmjota}, which determines the asymptotic expansion of $F^k(e)$ as $k\to \infty$. However,
  the present lemma is more elementary. 
  Alternatively, one can note that the operator $f:=\log\circ F\circ \exp$ is precisely in the class of operators considered in \cite[Section~4]{gg98a}.
The lemma is a special case of Theorem~29 there.
\end{proof}

Therefore, from now on, we make the following
assumption.
\begin{assumption}\label{assump-a}
The game is Despot-free and irreducible.
\end{assumption}

We also make the following assumption.
\begin{assumption}\label{assump-b}
The weights $m_{p,d}$ are integers. 
\end{assumption}
The case in which the weights are rational numbers
reduces to this one (multiplying all weights by a
common denominator does not change optimal positional strategies).

\subsection{Reduction to a well conditioned convex programming problem}
\label{subsec-cond}

Our strategy, to prove Theorem~\ref{th-polytime}
when the game is irreducible is to apply the ellipsoid method
to the convex programming 
formulation~\eqref{e-cvx0}. To do so, we must replace
this formulation by another
convex program whose feasible set 
is included in a ball $B_2(a,R)$,
(the Euclidean ball with center
$a$ and radius $R$),
and contains a Euclidean ball $B_2(a,r)$, where
$\log (R/r)$ is polynomially bounded 
in the size of the input. The following
key lemma allows us to do so. It bounds
the non-linear eigenvalue
and eigenvector of the dynamic programming
operator $F$, which have been shown to exist
in \Cref{e-valueunique}.
We set $\maxweight:=\max_{(p,d)\in E} m_{p,d}$
and $n:=|D|$. 
\begin{lemma}\label{lem-1}
Suppose the game is Despot-free and irreducible.
Then, the value $\lambda$ of the game is such that $1\leq \lambda \leq n\maxweight$.
Moreover, there exists a vector $U\in \operatorname{int}\R_+^n$ such that
$F(U)=\lambda U $, and 
\begin{align}\label{e-boundU}
1\leq U_d\leq \lambda^{n-1} \enspace , \forall d\in D \enspace .
\end{align}
\end{lemma}

\begin{proof}
(1) 
The fact that $\lambda \leq n\maxweight$ follows from the first Collatz-Wielandt
formula~\eqref{e-cw1}, which implies that $\lambda \leq \max_d F_d(e)$,
where $e$ is the unit vector of $\R^{D}$. We have $\max_d F_d(e)
 \leq n\maxweight$. Similarly, the last Collatz-Wielandt formula~\eqref{e-cwlast}
implies that $\lambda \geq \min_d F_d(e)$.
Since we assumed that every node in $G$ has at at least
one successor, we have $F_d(e)\geq 1$ for all $d\in D$,
and so $\lambda \geq 1$.

(2) Let $(d,d')$ be an arc in $\bar{G}$,
corresponding to a path of length $1$
in ${G}$, of the form $(d,t,p,d')$.
Then, $m_{pd'}U_{d'}\leq F_d(U) = \lambda U_d$ holds, 
with $\lambda \leq n\maxweight$ and $m_{pd'}$ integer. 
In particular, $U_{d'}\leq \lambda U_d$ holds.
Since the game is irreducible, any two vertices of $\bar{G}$
are connected by a path of length at most $n-1$.
It follows that $U_{d'}/U_d\leq \lambda^{n-1}$ holds for
all $d,d'\in D$. We may assume that the minimal
entry of $U$ is equal to $1$, by dividing $U$ by this
minimal entry. Then, $U_d\leq \lambda^{n-1}$ holds for all
$d$.
\end{proof}

We denote by $\sK$ the set of pairs $(u,\mu)\in \R^D\times \R\simeq \R^{n+1}$,
such that
\begin{subequations}\label{e-knew}
\begin{align}
f(u) & \leq \mu \unit + u \label{e-k1new}\\
0 & \leq u_d \leq (n-1)\ceil{\log(n\maxweight)}, \forall d \in D,\label{e-k2new}\\
0 & \leq \mu \leq \ceil{\log(n\maxweight)}+2 \enspace,\label{e-k3new}
\end{align}
\end{subequations}

where $\ceil{t}$ denotes the smallest integer greater than or equal to $t$,
and $f$ is given by~\eqref{e-shapley}, recalling that $\unit$ denotes the
unit vector of $\R^D$.
Recall that $W\geq 1$ since this is an integer, and that if $n\neq 1$, then
$(n-1)\ceil{\log(nW)}\geq 1$.
By combining Corollary~\ref{cor-cw}, Lemma~\ref{e-valueunique} and Lemma~\ref{lem-1}, we arrive at the following result.
\begin{proposition}\label{prop-compact}
The value of a Despot-free irreducible entropy game coincides with the exponential of the value of the convex program:
\begin{align}\label{e-sk0new}
\min \mu, \; (u,\mu)\in \sK \enspace ,
\end{align}
where $\sK$ is defined by~\eqref{e-knew}.
\end{proposition}
\begin{proof}
If $(u,\mu)\in\sK$, then it satisfies~\eqref{e-cvx0}, so 
by Corollary~\ref{cor-cw}, 
it is  not smaller than the logarithm of the value of the game.
Hence, the value of~\eqref{e-sk0new} is an upper
bound for the logarithm of the value of the game.
Now, if we take $\lambda$ to be equal to the value
of the game, we know by Lemma~\ref{lem-1}
and Lemma~\ref{e-valueunique} that $1\leq \lambda \leq nW$  and that we
can find a vector $U\in \operatorname{int}\R_+^n$
such that $F(U)=\lambda U$, and the bounds~\eqref{e-boundU}
on $U$ hold. Then, setting
$u:=\log(U)$ and $\mu:=\log \lambda$, 
we get that $(u,\mu)\in \sK$. It follows
that the exponential of the value of~\eqref{e-sk0new}
coincides with the value of the game.

Finally, the convexity of $\sK$ follows from the convexity of every coordinate
map of $f$, which is an immediate consequence of Lemma~\ref{lemma-rw}.
\end{proof}

We denote by $B_2(a,r)$ the Euclidean ball with center
$a$ and radius $r$. The sup-norm ball with the same
radius and center is denoted by $B_\infty(a,r)$.
We have the following lemma.
\begin{lemma}\label{lem-boundrR}
Let $a=((1/2) e,\ceil{\log(n\maxweight)}+3/2)\in \R^D\times \R$, and
let
\begin{align}\label{e-def-R}
r:= 1/3,
\qquad R := \sqrt{n+1}((n-1)\log (n\maxweight)+n+1) \enspace .
\end{align}
Then,  
\[
B_2(a,r) \subset \sK \subset B_2(a,R)  \enspace .
\]
\end{lemma}
\begin{proof}
Any point $(u,\mu)$ in $B_\infty(a,r)$ satisfies
$(1/2-r)e \leq u \leq (1/2+r)e$ 
and $\ceil{\log(n\maxweight)}+3/2-r\leq \mu \leq \ceil{\log(n\maxweight)}+3/2+r$. Since $n\neq 1$, we get that  $(n-1)\ceil{\log(nW)}\geq 1$,
and since $r\leq 1/2$, we obtain that $(u,\mu)$ satisfies the box
constraints~\eqref{e-k2new} and \eqref{e-k3new} defining $\sK$.
Since $f$ is order
preserving and commutes with the addition of a constant,
\begin{align*}
f(u) & \leq f((1/2+r)e)  = (1/2+r) e + f(0) \\
& \leq (1/2+r  + \ceil{\log(n\maxweight)})e \leq 
(1/2+r  + \ceil{\log(n\maxweight)})e 
+ (u+(r-1/2)e)\\
&= (2r+\ceil{\log(n\maxweight)} )e+u
 \leq (3r -1 +\mu)e +u 
\end{align*}
and so $f(u) \leq \mu e +u$ as soon as $r\leq 1/3$. 
We deduce
that $B_2(a,1/3)\subset B_\infty(a,1/3)\subset \sK$.

Moreover, since $\sK$ is included in a box of width $\ell=
(n-1)\ceil{\log(n\maxweight)}+2 $,
for any choice of $a'\in \sK$, $\sK$ is included in the sup-norm
ball $B_\infty(a',\ell)$,
and so, in the euclidean ball $B_2(a',\ell\sqrt{n+1})$. 
It follows that $\sK\subset B_2(a,R)$.
\end{proof}
\subsection{Construction of a polynomial time weak separation oracle}
We shall solve Problem~\eqref{e-sk0new} by the ellipsoid method~\cite{schrijver}.
The latter needs the following notions.
\begin{definition}\label{def-weak-separ}
Let $\sK$ denote a convex body in $\R^q$.
A {\em weak separation oracle} for $\sK$
is a procedure, taking as input a rational number $\nu>0$ and 
a rational vector $y\in \R^q$,
which concludes one of the following:
(i) asserting that $y$ is at Euclidean distance
at most $\nu$ from $\sK$; (ii) finding an 
{\em approximate separating half-space of precision $\nu$},
i.e., a linear form 
$\phi: x \mapsto 
c\cdot x$, with $c\in \R^q$, of Euclidean norm at least $1$,
such that
for every $x\in \sK$, 
\[ \phi(x) \leq \phi(y) + \nu
\enspace . 
\]
\end{definition}
Let us now recall the main complexity result about the ellipsoid method~\cite{schrijver}. 
To do so, we denote by $\bitsize r$ the number of bits needed to code an object $r$,
under the standard binary encoding. 
For instance, if $r$ is an integer, $\bitsize{r}:=\ceil{\log_2(r)}+1$,
if $r=p/q$ is a rational, $\bitsize{r}:=\bitsize{p}+\bitsize{q}$,
if $r=(r_i)$ is a rational vector, $\bitsize{r}:= \sum_i \bitsize{r_i}$,
and if $\psi$ is a linear form with rational coefficients over $\R^q$,
$\psi(x)=\sum_i r_i x_i$, for $x\in \R^q$,
then $\bitsize{\psi}=\sum_i \bitsize{r_i}$.
Here and after, the notion of length of an input refers
to the binary encoding. 

The ellipsoid method can be applied to solve the following
problem consisting in finding an approximate
minimum of precision $\epsilon$ of a linear
form $\psi$ with rational coefficients over a convex
body $\sK \subset \R^q$. This means looking for a vector
$x^*$ such that $d_2(x^*,\sK)\leq \epsilon$ and
$\psi(x^*) \leq \min_{x\in \sK} \psi(x) + \epsilon$,
where $d_2$ denotes the Euclidean distance.
We assume that we know a vector $a\in \sK$
with rational coordinates, and rational numbers $0<r<R$ such that 
\[
B(a,r)\subset \sK \subset B(a,R) \enspace .
\]
In that case, the size of the input of the approximate minimization problem
is measured by 
$\bitsize \psi + \bitsize a + \bitsize r + \bitsize R + \bitsize \epsilon$.

It is shown in~\cite{schrijver} that if the convex set
$\sK$ admits a polynomial time weak separation oracle,
the ellipsoid method computes an $\epsilon$-approximate
solution of the minimization problem in a time polynomial
in the size of the input. Specialized to the
present setting, and taking into account
the polynomial estimates for $\log r$ and $\log R$
in Lemma~\ref{lem-boundrR}, we get the following result.

\begin{theorem}[Corollary of~{\cite[Th.~3.1]{schrijver}}]
\label{th-compl-ellipsoid}
Suppose that the set $\sK$ defined by~\eqref{e-knew}
admits a weak separation oracle which
runs in polynomial time in the bitsize of the input
and in the bitsize of the game. Then, 
the ellipsoid method returns an approximate optimal solution
of precision $\epsilon$ of Problem~\eqref{e-sk0new},
i.e., a vector $(u,\mu)$ such that 
that $d_2((u,\mu),\sK)\leq \epsilon$
and $\mu$ does not exceed the value
of Problem~\eqref{e-sk0new} by more than $\epsilon$,
in a time that is polynomial in 
$\bitsize{\epsilon} + |E| + \log \maxweight$. 
\end{theorem}
Recall that $|\cdot|$ denotes the cardinality of a set, in particular
$|E|$ denotes the number of arcs of $G$.

The following result allows us to apply \Cref{th-compl-ellipsoid}
to Problem~\eqref{e-sk0new}.
\begin{proposition}\label{lem-separ}
The convex set $\sK$ defined by~\eqref{e-knew}
admits a weak 
separation oracle which runs in polynomial
time. 
\end{proposition}
To show this proposition, we need a series of arguments. Some of these
arguments, like the next lemma, are standard, whereas
other arguments require some transparent but rather technical bookkeeping,
exploiting the non-expansive character of $f$
to control the approximation errors. 

\begin{lemma}\label{lem-st}
Let $\epsilon>0$ and $t$ be rational numbers, and assume first that $t\leq 0$. 
Then, a rational approximation of absolute precision $\epsilon$ 
of $\exp(t)$ can be computed in a time that is polynomial 
in $\bitsize t$ and $\bitsize\epsilon$.
Assume now that $t>0$.  Then, 
a rational approximation of absolute precision $\epsilon$
of $\log(t)$ can be computed in a time that is polynomial 
in $\bitsize t$ and $\bitsize \epsilon$. 
\end{lemma}
\begin{proof}
It is shown in~\cite{borwein} that the conclusion
is true when the input belongs to a fixed compact subset of 
the intervals $(-\infty,0]$, in the case of $\exp$,
or of $(0,\infty)$, in the case of log.
The fact that the same property still holds for
the whole intervals $(-\infty,0]$ and $(0,\infty)$
follows from the range reduction techniques~\cite{muller}. 
\end{proof}

\begin{lemma}\label{lem-evalpoly}
Let $x$ be a vector in $\R^D$ with rational entries, and let $\epsilon>0$
be a rational number. 
An approximation of $f(x)$ with a sup-norm error not exceeding $\epsilon$
can be obtained in polynomial time in $\bitsize{x}+\bitsize{\epsilon}
+|E|+\log \maxweight$.
\end{lemma}
\begin{proof}
We have 
\[f_d(x) = \max_{(\sigma(d),p)\in E} \log(\sum_{(p,d')\in E}m_{pd'}\exp(x_{d'})) 
\enspace.
\]
Hence, it suffices to check that for every $p\in P$,
the value 
\[
h(x):=\log(\sum_{(p,d')\in E}m_{pd'}\exp(x_{d'})) 
\]
can be approximated with a precision $\epsilon$ within
a polynomial time. We set $\bar x:= \max_{d'\in D} x_{d'}$,
and make the change of variables $x_{d'} = \bar x+ \tilde{x}_{d'}$, so
that
\[
h(x) = \bar x + \log t,
\; t:= \sum_{(p,d')\in E} m_{pd'}\exp(\tilde{x}_{d'}) 
\]
We observe that $1\leq t$
and that $\log$ has Lipschitz constant $1$ over $[1,\infty)$. Hence,
to evaluate $h(x)$ with a precision $\epsilon$, it suffices to 
compute an approximation $\tilde{t}$ of $t$ with precision
$\epsilon/2$, which can be done in polynomial time
thanks to Lemma~\ref{lem-st}, and then to approximate $\log \tilde{t}$
with precision $\epsilon/2$, which can also
be done in polynomial time by the same lemma.
\end{proof}

\begin{proof}[Proof of Proposition~\ref{lem-separ}]
Let $\nu>0$. Our aim is to check whether a given pair
$(\bar v,\bar \mu)$ is at distance at most $\nu$
from $\sK$. Since we already showed
that we can get an %
 approximation
of $f$ in polynomial time, the proof will
be a matter of routine bookkeeping
(except perhaps the use of the subdifferential of $f$
to construct a separating halfspace). 

We denote by $\epsilon>0$ a rational number, $\epsilon\leq 1$, which we shall
fix in the course of the proof.

We provide the announced separation oracle. We first check
that every
box constraint, as well as the non-linear constraints $f_d(\bar v)\leq \bar v_d+ \bar \mu$, 
for $d\in D$, are satisfied up to a precision $\epsilon$, which can be
done in polynomial time in
 $\bitsize{\bar{v}}+\bitsize{\epsilon}
+|E|+\log \maxweight$ thanks to Lemma~\ref{lem-evalpoly}. 

(i) If these constraints are satisfied up to a precision
$\epsilon$, we have $-\epsilon \leq \bar v_d \leq (n-1)\ceil{\log (n\maxweight)} +\epsilon$,
$-\epsilon \leq \bar\mu \leq \ceil{\log (n\maxweight)}+2 + \epsilon$, and
$f(\bar v) \leq (\bar\mu+ \epsilon)e + \bar v$. 
Setting $\tilde{v}_d= \min(\max(\bar v_d,0),\ceil{(n-1)\log(n\maxweight)})$,
we get that $\|\bar v-\tilde{v}\|_\infty\leq \epsilon$
with $\tilde{v}$ satisfying the constraint~\eqref{e-k2new}.
Using the fact that $f$ is nonexpansive in the sup-norm, 
we deduce that 
$f(\tilde{v})\leq \tilde{\mu}e+ \tilde{v}$,
where $\tilde{\mu} =\bar  \mu + 3\epsilon$.
Moreover, $0\leq \tilde\mu \leq \ceil{\log (n\maxweight)}+2 + 4\epsilon$.
Now, since $f$ is also convex, for all $t\in [0,1]$,
we have
\[ f(t \tilde{v})\leq t f(\tilde{v})+(1-t) f(0)
\leq \tilde{\mu}' e+ t \tilde{v}\enspace ,\]
with  $\tilde{\mu}'=t\tilde{\mu}+(1-t) \ceil{\log (n\maxweight)}$.
Hence, taking $t=1/(1+2\epsilon)$, we get that $\tilde{\mu}'$ satisfies
the constraint~\eqref{e-k3new}, whereas $\tilde{v}'=t\tilde{v}$ still satisfies
the constraint~\eqref{e-k2new}, so that $(\tilde{v}',\tilde{\mu}')$ 
belongs to  $\sK$.
Using  $t\leq 2\epsilon$, we also have
$\|(\bar v,\bar \mu)-(\tilde{v}',\tilde{\mu}')\|_\infty \leq \epsilon L$, with
$L=(5+2 (n-1)\ceil{\log (n\maxweight)})$,
implying that $d_2((\bar v,\bar \mu),\sK) \leq L\sqrt{n+1}\epsilon$.
Hence, we shall require that
\[
\epsilon \leq \epsilon_1:= \frac{\nu}{(5+2 (n-1)\ceil{\log (n\maxweight)})
\sqrt{n+1}}\enspace ,
\]
to make sure that $d_2((\bar v,\bar \mu),\sK)\leq \nu$.

(ii) Assume now that one of the box constraints is violated by more than 
$\epsilon$. 
Then, one of the
linear forms $(v,\mu) \mapsto \pm v_d$ or $(v,\mu) \mapsto \pm \mu$
provides a separating half-space, and the norm of this
linear form is $1$. 
Assume finally that all the 
box constraints are satisfied up to $\epsilon$, and that 
one of the non-linear constraints
is violated by more than $\epsilon$. 
Let us write this constraint as 
\[
g(v,\mu) := \log(\sum_{d'\in D} m_{pd'} \exp(v_{d'})) - v_d - \mu\leq 0 
\]
for some $d\in D$ and $(\sigma(d),p)\in E$, 
so that $g(\bar v, \bar \mu)> \epsilon$. 
Since $g$ is convex,  the differential $\phi$ of $g$ at 
point $(\bar v, \bar \mu)$ satisfies
\[
\phi(v-\bar v, \mu - \bar\mu) 
\leq g(v,\mu) - g(\bar v,\bar\mu)\leq 
- \epsilon
\]
for all $(v,\mu) \in \sK$, i.e.,
\[
\phi(v,\mu) \leq \phi(\bar v,\bar\mu)-\epsilon,\;\forall (v,\mu)\in \sK
\]
showing that $\phi$ is a separating half-space.
However, we need an approximate
half-space given by a linear form with rational coefficients,
which we next construct by approximating  $\phi$.

To do so, we first compute the differential of $g$ at point $(\bar v,\bar \mu)$.
This is the linear form 
\[ 
\phi: (x,y)\in \R^D\times \R \mapsto \sum_{d'} x_{d'} m_{pd'} \exp(\bar v_{d'})/(\sum_{d''} m_{pd''}\exp(\bar v_{d''}))-x_d-y\enspace .
\]
The maximum of $\bar v$ can be subtracted to every coordinate of $\bar v$ 
without changing this linear form. Then, by Lemma~\ref{lem-st},
the coefficients of this linear form can be approximated in polynomial time. 
It follows that we can compute an approximation $\tilde{\phi}$
of $\phi$ of precision $\epsilon$ in polynomial time in 
$\bitsize {\bar v} +\bitsize {\epsilon}$.
Observe that the coefficient of the variable $y$ 
in the linear form $\phi$ is always
equal to $-1$. Hence, the approximate
linear form $\tilde{\phi}$ can be chosen with the same
coefficient, and then,
$\tilde{\phi}$ is of norm at least $1$.

Since any element $(v,\mu)$ of $\sK$
satisfies the box constraints~\eqref{e-k2new} and \eqref{e-k3new},
whereas $(\bar v, \bar \mu)$ satisfies these constraints up to 
$\epsilon\leq 1$, we get that 
\[ |(v,\mu)- (\bar v,\bar \mu)|_\infty
\leq M:=(n-1)\ceil{\log (n\maxweight)})+3\enspace ,\]
hence
\[
|\tilde{\phi}(v-\bar v,\mu-\bar \mu )- \phi(v-\bar v,\mu-\bar \mu)|
\leq M (D+1)\epsilon, \forall (v,\mu)\in \sK \enspace .
\]
So %
it suffices that
\[
\epsilon \leq \epsilon_2 := \frac{\nu}{M(D+1)}
\] 
to make sure that $\tilde{\phi}$ defines an approximate
separating half-space of precision $\nu$.

To summarize, it suffices to take $\epsilon=%
\min(\epsilon_1,\epsilon_2)$
in the previous analysis, so that the conditions
of Definition~\ref{def-weak-separ} are satisfied. Moreover,
for this choice, %
all the computations take a polynomial time in
$\bitsize{\nu}+\bitsize{\bar v}+\bitsize{\bar\mu}$,
and the size $|E|+\log W$ of the description of the game. 
\end{proof}
\subsection{Using separation bounds between algebraic numbers to compare policies}
It follows from Theorem~\ref{th-compl-ellipsoid}
that we can compute in polynomial time an approximate
solution of Problem~\eqref{e-sk0new} with precision
$\epsilon$. We next show that it is possible to choose
$\epsilon$ with a polynomial number of bits, in such a way
that this approximate solution allows us to identify
an optimal policy. 
We shall actually prove a version of this result in the more general two-player case.
This extended version, stated as \Cref{cor-separation}, will apply both to the Despot-free case,
and to the case of entropy games with a fixed number of
significant states, see Section~\ref{sec-deriva}.
To prove it, %
we rely on separation
bounds for algebraic numbers. 
\begin{theorem}[\cite{rump}]\label{th-rump}
Let $p$ be a univariate polynomial of degree $n$
with integer coefficients, possibly
with multiple roots. Let $S$
be the sum of the absolute values of its coefficients. Then, 
the distance between any two distinct roots of $p$ is at least
\[
(2 n^{\frac n2+2}(S+1)^n)^{-1} \enspace .
\]
\end{theorem}
To any given pair $(\delta,\tau)$ of policies, one can associate
a directed sub-graph $G(\delta,\tau)$ of $G$, obtained,
by erasing for all $d\in D$, every arc in $\{(d,t)\in E\}$
except $\{d,\delta(d)\}$, and similarly, by
erasing for all $t\in T$, every arc in $\{(t,p)\in E\}$
except $\{t,\tau(t)\}$. The dynamic programming operator of the
game associated to this  sub-graph $G(\delta,\tau)$ coincides
with the conjugate $\Fdt{\delta}{\tau}:=\exp\circ {}^{\tau}f^{\delta} \circ \log$
of~\eqref{e-Mdeltatau}
and is equal to the linear operator with matrix 
$\Mdt{\delta}{\tau}\in \mathbb{N}^{D \times D}$:
\begin{equation}\label{defFtaudel}
\Fdt{\delta}{\tau}:\R^D\to\R^D,\; X\mapsto \Mdt{\delta}{\tau} X
\enspace, 
\end{equation} where the entry $(d,d')$ of 
$\Mdt{\delta}{\tau}$ is equal to 
$m_{\tau\circ \delta(d),d'}$ when $(\tau\circ \delta(d),d')\in E$ and to
zero otherwise.
Then, the value of the entropy 
game starting in state $d$, $R^\infty_d(\delta,\tau)$
coincides with the maximum of the
Perron-roots (that is the positive eigenvalues which coincide with
the spectral radii)
 of the submatrices of ${}^{\tau}M^{\delta}$
with nodes in a strongly connected component to which $d$ has access
in the graph $G(\delta,\tau)$. The value of the original entropy game coincides
with the Perron-root of ${}^{\tau}M^{\delta}$.

\begin{corollary}\label{cor-separation}
There exists
a rational function $(n,\maxweight)\mapsto \eta_{\text{sep}}(n,\maxweight)>0$ such that for every
 two different pairs  $(\delta,\tau)$ and  $(\delta',\tau')$ which yield different values of the entropy game,  these two values differ by at least $\eta_{\text{sep}}(n,\maxweight)$ and 
\[
\eta_{\text{sep}}(n,\maxweight) \geq \exp(-\text{poly}(n+\log \maxweight))
\]
where the polynomial inside the exponential is
independent of the input.
\end{corollary}
\begin{proof}
Given a pair $(\delta,\tau)$ of policies, the values of the entropy game
are eigenvalues of the matrix $A={}^{\tau}M^{\delta}\in \mathbb{N}^{D \times D}$.
Observe that the entries of $A$ are integers bounded by $\maxweight$.
Let $f_A$ be the characteristic polynomial of $A$. 
The coefficient of the monomial of degree
$n-k$ in  $f_A$ is the sum of the $C_n^k$ principal
minors of $A$ of size $k$. By Hadamard's inequality, each 
absolute value of these minors
is at most $(\sqrt{n}\maxweight)^k$, and so, every coefficient of $f_A$
has an absolute value 
bounded by $C_n^k(\sqrt{n}\maxweight)^k$ and their sum is 
$\leq (2\sqrt{n}\maxweight)^n$. 
Two different pairs of strategies yield two characteristic
polynomials, $f_A$ and $f_B$, whose product is of degree
$2n$ and whose sum of absolute value of
coefficients is bounded by the product of such bounds for
$f_A$ and $f_B$, so by $(2\sqrt{n}\maxweight)^{2n}$. 
Therefore, the size $S$ appearing in Theorem~\ref{th-rump}
is bounded by $(2\sqrt{n}\maxweight)^{2n}$. We deduce that 
if the two pairs of strategies yield distinct values,
the distance between these values is at least
\[
\eta_{\text{sep}}(n,\maxweight) :=(2 (2n)^{n+2}((2\ceil{\sqrt{n}}\maxweight)^{2n}+1)^{2n})^{-1} \enspace .
\]
This number is rational and satisfies
\[
\eta_{\text{sep}}(n,\maxweight) \geq \exp(-\text{poly}(n+\log \maxweight)),
\]
for some polynomial function $\text{poly}$.
Since the above lower bound is true for every two pairs of different policies $(\delta,\tau)$ and $(\delta',\tau')$,
we obtain the result.
\end{proof}
Hence, if two policies of Tribune yield different values $\lambda$ and $\lambda'$, then, $|\lambda-\lambda'|$ is bounded below by 
the rational number $\etasep>0$ whose number of bits is polynomially bounded
in the size of the input.

\if{
The convex programming formulation 
allows us to compute the logarithm
of the value, instead of the value.
Hence, we need a separation bound for these
values.
\begin{lemma}
If two different pairs  $(\delta,\tau)$ and  $(\delta',\tau')$ of policies in an entropy game yield two different value $\lambda$ and $\lambda'$.
Then, there exists
a function $\eta^{\log}_{\text{sep}}(n,\maxweight)$ taking positive values, such that 
$\log \lambda$ and $\log \lambda'$
 two values differ by at least $\eta^{\log}_{\text{sep}}(n,\maxweight)$ and 
\[
\eta^{\log}_{\text{sep}}(n,\maxweight) \geq \exp(-\text{poly}(n+\log \maxweight))
\]
\end{lemma}
\begin{proof}
The values $\lambda,\lambda'$ are both
contained in the interval $[1,n\maxweight]$,
and so their logarithms are in $I:=[0,\log(n\maxweight)]$. 
The exponential function is of Lipschitz constant $n\maxweight$
on the interval $I$. Hence, we can take
\[
\eta^{\log}_{\text{sep}}(n,\maxweight) \leq \frac{\eta_{\text{sep}}(n,\maxweight)}{n\maxweight}
\]
which has still the announced polynomial growth.
\end{proof}}\fi
\if{
We now take $u,\mu$ be an approximate solution with rational coordinates
of Problem~\eqref{e-sk0new} of
precision $\epsilon$, meaning
that $d((u,\mu),\sK)\leq \epsilon$
and that $\mu \leq \log\lambda^*+\epsilon$,
where $\lambda^*$ is the value of the entropy game.

For each $t\in T$ and $(t,p)\in E$, we compute
$\tilde{w}_{tp} $, a rational number that approximates
with precision $\epsilon$ the expression
\[
\log(\sum_{d'\in D} m_{p,d'}\exp(u_{d'})) 
\]
and define the policy $\tau$ 
such that $\tau(t)$ is an index $p$ attaining
the maximum in $\max_{(t,p)\in E} \tilde{w}_{tp}$.

\begin{lemma}
If $2\epsilon < \eta^{\log}_{\text{sep}}$, then, the policy
\tau$ which we just constructed is an optimal
strategy of Tribune.
\end{lemma}
\begin{proof}
\todo[inline]{SG: the bug is in the next inequality. We can only show that an inequality of this type can be satisfied for the indices of a basic class of an optimal matrix, and to do so we need to use a complementary slackness/duality argument and bound the left eigenvector $m$ of the associated matrix. This is not so difficult because $m_i/m_j \leq$ bound of the same style as for $U$. The patch consists in showing that lower inequality is OK on the set of indices belonging to the union $C$ of basic classes, and that the control can be chosen so that the weight of the $C \times $ complement C block of the matrix is $<<\epsilon$}
The strategy satisfies
\[
f^{\delta}(u) \geq u + (\log \lambda^*-2\epsilon)
\]
It follows that the value $\lambda'$ of $f^{\delta}$
is at least $\log \lambda^*-2\epsilon$. For $2\epsilon<\eta^{\log}_{\text{sep}}$,
this is impossible unless $\lambda'=\lambda^*$. We conclude that $\delta^*$
is an optimal positional strategy. 
\end{proof}
This proves
}\fi

\subsection{Synthesis of an optimal strategy of Tribune
from an approximate solution of the convex program in Proposition~\ref{prop-compact}.}
\label{sec-synth-irred}

To any policy $\tau$ of Tribune, we associate 
a dynamic programming operator $\Ftau$,
which is the specialization of the map $\Fdt{\delta}{\tau}$ 
of previous section to the case where
$\delta=\sigma$. This is the
self-map of $\R^{D}$ defined by 
\[
\Ftau_{d}(X) = 
\sum_{(\tau(\sigma(d)),d') \in E} m_{\tau(\sigma(d))d'}X_{d'} 
\enspace.
\]
So $\Ftau(X)=\Mt{\tau} X$, where $\Mt{\tau}={}^{\tau}M^{\sigma}$
is the  $|D|\times |D|$ matrix with nonnegative entries
equal to $m_{\tau(\sigma(d))d'}$ when $(\tau(\sigma(d)),d')\in E$
and zero otherwise.
\subsubsection{The simpler situation in which every policy $\tau$
of Tribune yields an irreducible matrix}\label{subsec-simple}
To explain our method, we make first the restrictive
assumption that for every policy $\tau$ of Tribune,
the matrix $\Mt{\tau}$ is irreducible.
In particular, we can take an optimal policy $\tau^*$.
By a standard result of Perron-Frobenius theory~\cite{berman},
$\Mt{\tau^*}$ has 
a left eigenvector $\pi$ with positive entries,
associated to the spectral radius $\lambda^{\tau^*}:=\rho(\Mt{\tau^*})$,
called Perron root. Hence, $\pi \Mt{\tau^*}=\lambda^{\tau^*}\pi$.
Since $\tau^*$ is optimal, $\lambda^{\tau^*}=\lambda^*$,
where $\lambda^*$ is the value of the entropy game starting from
any node $d\in D$, see Lemma~\ref{e-valueunique}.
Moreover, by applying Lemma~\ref{lem-1} to the linear
map $U \mapsto (\Mt{\tau^*})^T U$, where $^T$ denotes
the transposition, we deduce that
$\pi_d/\pi_{d'}\leq (nW)^{n-1}$.

For any rational number $\epsilon>0$, the ellipsoid algorithm, applied to the optimization problem of Proposition~\ref{prop-compact}, 
yields in polynomial time a vector $u$ and a scalar $\mu$
such that $\mu\leq \log  \lambda^* +\epsilon$ and
$d_2((u,\mu),\sK)\leq \epsilon$.
So there exists  $(\tilde u,\tilde \mu)\in \sK$ such that
$\|u-\tilde{u}\|_\infty\leq \epsilon$ and $|\mu-\tilde \mu|\leq \epsilon$.
Since $(\tilde u,\tilde \mu)\in \sK$, and $\log \lambda^*$ is the
value of~\eqref{e-sk0new}, we deduce that
$\log \lambda^* \leq \tilde \mu\leq \mu +\epsilon$,
so  $\lambda^*  \exp(-\epsilon)\leq \exp(\mu) \leq \lambda^* \exp(\epsilon)$.
Using~\eqref{e-knew}, and assuming that $\epsilon\leq 1$, we deduce that
$u_d-u_{d'}\leq (n-1)\ceil{\log(n\maxweight)}+2\epsilon
\leq  (n-1)(\log(n\maxweight)+1)+2$, for all $d,d' \in D$.
Using the nonexpansivity of $f$, we also obtain that
$f(u)\leq f(\tilde{u})+\epsilon \unit 
\leq \tilde u +( \tilde \mu+\epsilon )\unit
\leq (\mu +3\epsilon)\unit + u
\leq  (\log \lambda^* +4\epsilon)\unit + u$.
Taking $U:=(U_d)_{d\in D}$ with $U_d := \exp(u_d)$, we get $F(U)\leq \lambda^* \exp(4\epsilon) U $ and $U_d/U_{d'}\leq (enW)^{n-1}e^2$.

We choose any 
policy  $\underline{\tau}$ such that $F(U)=\Mt{\underline{\tau}}U$.
Therefore, 
\[
\underline{\tau}(\sigma(d)) \in 
\argmax_{\tau \in \cP_{T}} \sum_{(\tau(\sigma(d)),d') \in E} m_{\tau(\sigma(d))d'}U_{d'} 
\enspace.
\]
 We claim that $\underline{\tau}$ is optimal
if $\epsilon$ is sufficiently small. 

To show the latter claim, we observe that $\Mt{\tau^*} U \leq F(U)$. 
For all $d\in D$,
\begin{align}
0 &\leq 
\pi_d (\lambda^*\exp(4\epsilon)U_d - F_d(U))\leq 
\pi_d (\lambda^*\exp(4\epsilon)U_d - (\Mt{\tau^*}U)_d)\nonumber\\
&\leq 
\sum_{d'\in D} \pi_{d'} (\lambda^*\exp(4\epsilon) U_{d'} - (\Mt{\tau^*}U)_{d'}) \nonumber\\
&= \pi (\lambda^*\exp(4\epsilon) U - \Mt{\tau^*}U)  \nonumber\\
&= \lambda^*(\exp(4\epsilon)-1) \pi U\enspace. 
\label{e-cs}
\end{align}
Using $\pi_d/\pi_{d'}\leq (nW)^{n-1}$ and $U_d/U_{d'}\leq (enW)^{n-1}e^2$,
we deduce that 
$\pi U\leq \pi_d U_d (1+(n-1) (enW)^{2(n-1)}e)$, so
$F(U) \geq \underline{\lambda} U$, where
$\underline{\lambda}:= \lambda^*[1- (\exp(4\epsilon) -1) 
(n-1) e (enW)^{2(n-1)}]$.
In view of the formula of $\underline{\lambda}$,
we can choose $\epsilon>0$, with a polynomially
bounded number of bits, such that $\underline{\lambda}> \lambda^* - \etasep$.
Since, $\Mt{\underline{\tau}}U\geq \underline{\lambda} U$,
we have $\rho(\Mt{\underline{\tau}})\geq \underline{\lambda}$
and so $\rho(\Mt{\underline{\tau}})> \lambda^* - \etasep$.
Since $\lambda^*$ is the maximum of the values
of all the policies,  $\rho(\Mt{\underline{\tau}})\leq \lambda^*$.
By definition of the separation parameter $\etasep$ given in\Cref{cor-separation}, this implies
that $\rho(\Mt{\underline{\tau}})=\lambda^*$, and 
so the policy $\underline{\tau}$ of Tribune which we just constructed is optimal,
showing the claim. 

In the preceding argument, the computation~\eqref{e-cs} may look a bit magic at the first sight, it should become intuitive if one interprets it as an approximate complementary slackness condition for the semi-infinite program of
\Cref{prop-compact}, the invariant measure $\pi$ playing the role of a Lagrange multiplier.

When some policies $\tau$ yield a {\em reducible} matrix $\Mt{\tau}$, the synthesis
of the optimal policy $\underline{\tau}$ still exploits the same idea 
with an additional technicality, since we can only guarantee that the inequality
$F_d(U) \geq \underline{\lambda} U_d$ is valid
for every state $d$ such that $\pi_d>0$. We explain the more technical argument
in the next section. 

\subsubsection{Synthesis of an optimal strategy of Tribune, in general}
\label{sec-general}
\ 
\todo[inline]{MA: il fallait preciser le resultat sur l'existence de 
$B$ et $\pi$. Et ensuite la preuve etait incomprehensible et fausse.
J'ai change $S$ par $B$ dans la suite et refait une preuve plus
proche de celle de la section precedente.}
Recall that if $M$ is a reducible nonnegative matrix,
a {\em class} of $M$ is a strongly connected component of the directed graph
of $M$, and that this class is {\em basic} if the $B\times B$
submatrix of $M$, denoted by $M_{BB}$, has Perron root $\rho(M)$.
It is known~\cite{berman} that $M$ has always a basic class,
and a nonnegative left eigenvector associated with $\rho(M)$.
Moreover, choosing a basic class $B$ which is final among the basic classes
of $M$, that is such that the set $S$
of nodes $d'\in D$ that are reachable in the directed graph of $M$ starting
from some node in the basic class $B$ does not contain any node of
another basic class, then 
there exists a nonnegative left eigenvector  $\pi$
so that its support $\{d\mid \pi_d\neq 0\}$ coincides with $S$.
We shall assume that $\pi$,
$S$ and $B$ satisfy these properties, for
$M=\Mt{\tau^*}$ corresponding to an optimal policy $\tau^*$.
We set $N:=D\setminus B$, and for any $D\times D$ matrix $M$,
any vector $u\in \R^D$, and any subsets $F$ and $G$ of $D$,
we denote by $M_{FG}$ the $F\times G$ submatrix of $M$
and by $v_F$ the vector of $\R^F$ given by
$v_F:=(v_d)_{d\in F}$.
Since $B$ is a basic class and $B$ has access to any element 
of $S$, we get that no element of $S\setminus B$ has access to an element of
$B$, and since $\pi$ equals zero outside $S$,
we get that the restriction of $\pi \Mt{\tau^*}$ to
$B$ equals $\pi_B \Mt{\tau^*}_{BB}$ and so $\pi_B \Mt{\tau^*}_{BB}=
\lambda^* \pi_B $. 
The same computation as in \Cref{subsec-simple} restricted to
the elements $d\in B$ now gives
\begin{align*}
0 &\leq 
\pi_d (\lambda^*\exp(4\epsilon)U_d - F_d(U))
\leq \pi_d (\lambda^*\exp(4\epsilon)U_d - (\Mt{\tau^*}U)_d)\\
& \leq \pi_B(\lambda^*\exp(4\epsilon)U_B -
(\Mt{\tau^*}_{BB}U_B+\Mt{\tau^*}_{BN}U_N))\\
&= \lambda^*(\exp(4\epsilon)-1)\pi_B U_B
-  \pi_B\Mt{\tau^*}_{BN}U_N
\enspace .
\end{align*}
The bounds on $U$ obtained in \Cref{subsec-simple} are still
valid, and the ones of $\pi_d/\pi_{d'}$ are valid only for $d,d'\in B$
using $\pi_B \Mt{\tau^*}_{BB}=\lambda^* \pi_B $ and the irreducibility of
$\Mt{\tau^*}_{BB}$.
We deduce that $F_d(U) \geq \underline{\lambda} U_d$ for all $d\in B$,
for the same  $ \underline{\lambda} $ as in  \Cref{subsec-simple}.
Moreover, $\pi_B \Mt{\tau^*}_{BN}U_N \leq \lambda^*(\exp(4\epsilon)-1)\pi_BU_B$.
We define $\epsilon':= \lambda^*(\exp(4\epsilon)-1) e n(enW)^{2(n-1)}$, 
so that $\Mt{\tau^*}_{dN}U_N\leq \epsilon' U_d$ for all $d\in B$, where 
$\Mt{\tau^*}_{dN}$ is the $d$th line of the matrix $\Mt{\tau^*}_{BN}.$

\todo[inline]{MA: j'ai corrige ci-dessous en donnant un algo
pour trouver $\underline{\tau}$ sans connaitre $B$ et $N$ a l'avance.}
We first choose, any 
policy  $\underline{\tau}$ and set $B'$ such that 
$\Mt{\underline{\tau}}_{dD} U \geq \underline{\lambda} U_d$
and $\Mt{\underline{\tau}}_{dN'}U_{N'} \leq \epsilon' U_d$, for all $d\in B'$,
with $N'=D\setminus B'$.
We know from the above analysis that there
is always at least one policy $\underline{\tau}$ and set $B'$ 
with this property (namely $\tau^*$ and $B'=B$).
Moreover, such a policy and set can be obtained by the following
algorithm.
Indeed, let us start from any policy $\underline{\tau}$ such that
$F(U)=\Mt{\underline{\tau}}U$, and 
choose $B'$ as the set of $d\in D$ such that 
$\Mt{\underline{\tau}}_{dD} U \geq \underline{\lambda} U_d$.
Then, $B\subset B'$ since $\Mt{\underline{\tau}}_{dD} U=F_d(U)$.
At each step of the algorithm, one applies the following
operations to each $d\in B'$ in some order:
set $N'=D\setminus B'$ and
check if  $\Mt{\underline{\tau}}_{dN'}U_{N'} \leq \epsilon'$.
If this does not hold, change $\tau(\sigma(d))$ to any action
so that $\Mt{\underline{\tau}}_{dD} U \geq \underline{\lambda} U_d$
and $\Mt{\underline{\tau}}_{dN}U_N \leq \epsilon'$.
If this is impossible, then eliminate $d$ from $B'$ 
and continue.
Then stop at any step in which $B'$ does not change.
Since the cardinality of $B'$ decreases by one at each step to which one
does not stop and $B'\supset B$, we get that the algorithm
stops after at most $n$ iterations and needs at most $n^2$ 
products of a matrix by a vector, so it takes a polynomial
time. Moreover at each step and so at the end of the algorithm,
we have $B'\supset B$ and $N'\subset N$.

We deduce that $\Mt{\underline{\tau}}_{B'B'}U_{B'} 
\geq (\underline{\lambda}-\epsilon')U_{B'}$, showing 
that $\rho(\Mt{\underline{\tau}})\geq \rho(\Mt{\underline{\tau}}_{B'B'})
 \geq (\underline{\lambda}-\epsilon')$.
In view of the formula of $\underline{\lambda}$ and $\epsilon'$ 
we can always choose $\epsilon>0$, with a polynomially
bounded number of bits, such that 
$\underline{\lambda}-\epsilon'> \lambda^*-\etasep$.
Hence, 
$\rho(\Mt{\underline{\tau}})=\rho(\Mt{\underline{\tau}}_{B'B'})= \lambda^*$, 
since $\rho(\Mt{\underline{\tau}})$ is an eigenvalue of 
$\Mt{\underline{\tau}}$. We also deduce
from $\Mt{\underline{\tau}}_{B'B'}U_{B'} 
\geq (\underline{\lambda}-\epsilon')U_{B'}$ that every state $d\in B'$ 
has value $\lambda^*$. Finally, 
since the game is irreducible, we can always
replace $\underline{\tau}(\sigma(d))$ for $d\not \in B'$ to make $B'$ accessible from any initial state, so that the policy $\underline{\tau}$
is optimal.  This concludes the proof of \Cref{th-polytime}. 

\if{
\subsection{Reducible Despot-free entropy games}
When the directed graph $\bar G$ which we constructed has several strongly
connected components, $C_1,\dots, C_q$, we associate to
every such component $C_i$ a reduced entropy game, in which we
keep only the states $d\in C_i$, and only the arcs $(p,d')$
with $d'\in C_i$. We denote by $V(\infty,C_i)$ the value
of the reduced entropy game. 
\todo[inline]{Is the  game well defined? modify the definition
to handle states without successors? }
\begin{theorem}\label{th-downstream}
The value $V_d(\infty)$ coincides with the maximum
of the $V(\infty,C_i)$, taken over all the strongly
connected components to which $d$ has access in $\bar{G}$.
\end{theorem}
\begin{proof}
\end{proof}
Hence, to solve a reducible Despot-free entropy game,
we can compute first every value $V(\infty,C_i)$, by the
method of Section~\ref{subsec-irred},
and finally apply Theorem~\ref{th-downstream},
}\fi

\subsection{Derivation of Theorem~{\ref{cor-polytime}} from Theorem~{\ref{th-polytime}}}\label{sec-deriva}
By \Cref{cor-morphism}, 
\[
V^\infty_d = \min_{\delta \in \cP_D} V^\infty_d(\delta,\ast)
\]
Observe that $|\cP_D|\leq |E|^s$,
where $s$ is the number of significant states
for despot, hence, if $s$ is fixed,
this minimum involves a polynomial number of terms. 
Thanks to the separation bound given in \Cref{cor-separation}, it suffices
to compute an approximation of each $V^\infty_d(\delta,\ast)$
for some $\epsilon>0$ such that $\log \epsilon$ is polynomially
bounded in the size of the input, to make sure that
a policy $\delta$ achieving the minimum in the above
expression, in which every term $V^\infty_d(\delta,\ast)$
is replaced by its approximate value, is an optimal
policy.

\section{Multiplicative policy iteration algorithm and comparison with the spectral simplex method of Protasov}\label{sec-algopo}

We now consider the question of solving 
entropy games in practice.

\subsection{Algorithms}
The equivalence 
between entropy games and some special class of
stochastic mean payoff games, through logarithmic glasses
(see Section~\ref{sec-equivalence}), will allow us to adapt
classical algorithms for one or two player zero sum games,
such as the value iteration and the policy iteration algorithm.
We next present a multiplicative version of the policy iteration
algorithm, which follows by adapting policy iteration ideas
for two player games by Hoffman and Karp~\cite{HoffmanKarp},
with ``multiplicative'' policy iteration techniques
of Howard and Matheson~\cite{Howard-Matheson},
Rothblum~\cite{rothblum} and Sladky~\cite{Sladky1976},
The latter ``multiplicative'' policy iteration techniques
apply to the Despot-free case.   
For clarity, we shall explain first policy iteration in the special Despot-free case: this is more transparent, and this also will allow us to interpret Protasov's spectral simplex method~\cite{protasov} as a variant of policy iteration.
The newer part here is the two player case,
which is dealt with in \Cref{algo-jeux}.

We assume that $D=T=\{1,\ldots, |T|\}$ and $\sigma$ is the identity in \eqref{e-cvx0} .
Let $\Ftau$ and $\Mt{\tau}$, $\tau \in \cP_{T}$, be defined as in the 
previous section.
If $\Mt{\tau}$ is irreducible, in particular if all its entries are 
positive, $\Mt{\tau}$ has an eigenvector $X^\tau>0$,  associated to
the Perron root $\lambda^\tau:=\rho(\Mt{\tau})$.
Moreover,  $X^\tau$ is unique up to a multiplicative constant and is
called a Perron eigenvector.
If  all the matrices $\Mt{\tau}$, $\tau\in \cP_{T}$ are irreducible,
one can construct a multiplicative version of 
the policy iteration, \Cref{algopo1p}.

\begin{algorithm}%
\caption{Multiplicative policy Iteration for Despot-free entropy games}\label{algopo1p}
\begin{algorithmic}[1]

\State Initialize $k=1$, $\tau^{0}$, $\tau^{1} \neq \tau^{0}$ randomly.
\While{$\tau^{k} \neq \tau^{k-1} $}

\State \label{step3}
Compute the Perron root $\lambda^{\tau^{k}}$ and a Perron eigenvector $X^{\tau^{k}}$ of $\Mt{\tau^{k}}$.

\State Compute a new policy $\tau^{k+1}$ such that, for all $t \in T$,
     
$$\tau^{k+1}(t) \in 
\argmax_{p\in P,\; (t,p) \in E} \; \sum_{t'\in T,\; (p,t')\in E}m_{p,t'}X^{\tau^{k}}_{t'}\enspace ,$$
   and set $\tau^{k+1}(t) = \tau^{k}(t) $ if this choice is compatible with the former condition.
\State $k \gets k + 1 $
\EndWhile

\State \textbf{return} 
the optimal policy $\tau^{k}$,
the Perron root $\lambda^{\tau^k}$ and Perron eigenvector $X^{\tau^k}$ of $\Mt{\tau^{k}}$.
\end{algorithmic}
\end{algorithm}

The following result shows that Algorithm~\ref{algopo1p}
does terminate. The proof relies on a multiplicative version
of the classical strict monotonicity argument in policy iteration,
which was already used by
Howard and Matheson~\cite{Howard-Matheson},
Rothblum~\cite{rothblum} and Sladky~\cite{Sladky1976}.
We reproduce the short proof for completeness.
\begin{proposition}\label{conv-policy} Consider  Algorithm~\ref{algopo1p}, where 
the computations are performed in {\em exact arithmetics}
and all the matrices $\Mt{\tau}$, $\tau\in \cP_{T}$ are supposed to be irreducible.
Then, the  sequence $\lambda^{\tau^k}$ is increasing 
as long as $\tau^{k}\neq \tau^{k-1}$. Moreover, the algorithm ends after 
a finite number of iterations $k$, and $\lambda^{\tau^k}$ is the value of
the game (at any initial state).
\end{proposition}

\begin{proof}%
The property that  $\lambda^{\tau^k}$ is increasing 
uses the general property that for an irreducible
nonnegative matrix $M$, and for a positive vector $u$, $Mu\geq \lambda u$ (component-wise)
and $Mu\neq \lambda u$ implies $\rho(M)>\lambda$. 
Then, the algorithm terminates since the number
of policies $\tau$ is a finite set,
and so $\{ \lambda^{\tau} \; | \; \tau \in \cP_{T} \}$ 
is finite.
When $\tau^{k}=\tau^{k-1}$, we get $F(X^{\tau^{k}})=\lambda^{\tau^k} X^{\tau^{k}}$,
and Lemma~\ref{e-valueunique} shows that $\lambda^{\tau^k}$ is the value of
the game (at any initial state).
\end{proof}

\Cref{algopo1p} has a dual version, in which maximization is
replaced by minimization, in order to solve the Tribune-free setting of entropy games.
For this dual version of Algorithm~\ref{algopo1p},
the sequence  $\lambda^{\tau^k}$ is  decreasing as long as
$\tau^{k}\neq \tau^{k-1}$.
This uses the property (dual to the previous one) that for an irreducible
nonnegative matrix $M$, and for a positive vector $u$, $Mu\leq \lambda u$
and $Mu\neq \lambda u$ implies $\rho(M)<\lambda$. 
Then the algorithm terminates as for the primal version.

In practice, \Cref{algopo1p} can only be implemented in an approximate
way. A bottleneck in this algorithm is the computation
of the Perron root and Perron eigenvector. The later
can be computed by standard double
precision algorithms, like the QR method. The
latter method requires $O(D^3)$ flops, where $D$ is
the size of the matrix $^{\tau}M$ associated
to a policy $\tau$. (Note that such complexity estimates
in an informal ``floating point'' arithmetic model
are meaningful only for well conditioned instances,
in contrast with the Turing-model complexity estimates
that we derived unconditionnally in \Cref{sec-proof}.)
One may also use a more scalable algorithm,
like the power algorithm. In fact,
we shall see in \Cref{algoPower2p} that the power
idea can be applied directly and in a simpler way to solve the non-linear
equation, avoiding the recourse to policy iteration.
So it is not clear that the power algorithm will be
the best choice to compute the eigenpair in situations
in which \Cref{algopo1p} is competitive. In the experiments
which follow, we used the QR method in \Cref{algopo1p}.

In~\cite{protasov}, Protasov introduced the Spectral Simplex Algorithm.
His algorithm is a variant of Algorithm~\ref{algopo1p}
in which at every iteration the policy is improved only at {\em one} state,
which is the first state $t$ such that
$F_t(X^{\tau^k})> \lambda^{\tau^k}X^{\tau^k}_t$.
We shall also consider another version of Algorithm~\ref{algopo1p},
in which we  also change the policy at only one state $t$, but we choose it in order to maximize the expression $F_t(X^{\tau^k})- \lambda^{\tau^k}X^{\tau^k}_t$. 
We shall refer to this algorithm  as "Spectral Simplex-D"
since this is analogous
to Dantzig's pivot rule in the original simplex method~\cite{ye2010simplex}.

The Spectral Simplex Algorithm introduced by Protasov in~\cite{protasov}
is described in the Despot-free setting in Algorithm~\ref{algossp}.

\begin{algorithm}%
\caption{Spectral Simplex Algorithm \cite{protasov}}\label{algossp}
\begin{algorithmic}[1]
\State Initialize $k=1$, $\tau^{0}$, $\tau^{1} \neq \tau^{0}$ randomly.
\While{$\tau^{k} \neq \tau^{k-1} $}
        \State Compute the Perron root $\lambda^{\tau^{k}}$ and a Perron eigenvector $X^{\tau^{k}}$ of $\Mt{\tau^{k}}$.
        \State Set $\tau^{k+1}(t)= \tau^{k}(t)$ for $t\in T=\{1,\dots,|T|\}$.
	\State bool = true, $t=1$.
	\While{(bool) and ($t \leq |T|$)}
		\If{ $\Ftauk_{t}(X^{\tau^{k}})\neq \max_{p\in P,\; (t,p) \in E} \; \sum_{t'\in T,\; (p,t')\in E}m_{p,t'}X^{\tau^{k}}_{t'}$}
		\State Choose 
$$\tau^{k+1}(t) \in 
\argmax_{p\in P,\; (t,p) \in E} \; \sum_{t'\in T,\; (p,t')\in E}m_{p,t'}X^{\tau^{k}}_{t'}\enspace .$$
		\State bool = false.
		\EndIf
		\State $t \gets t+1$.
	\EndWhile
\State $k \gets k + 1 $
\EndWhile

\State \textbf{return} the first policy $\tau^{k}$ such that $\tau^{k-1} = \tau^{k}$, the Perron root $\lambda^{\tau^k}$ and a Perron eigenvector $X^{\tau^k}$ of the matrix $\Mt{\tau^{k}}$.

\end{algorithmic}
\end{algorithm}

The Spectral Simplex Algorithm with Dantzig update is
shown in \Cref{algossp:DF-Dz}, again in the Despot-free setting.

\begin{algorithm}%
\caption{Spectral Simplex Algorithm for Despot free game - Dantzig update}\label{algossp:DF-Dz}
\begin{algorithmic}[1]
\State Initialize $k=1$, $\tau^{0}$, $\tau^{1} \neq \tau^{0}$ randomly.
\While{$\tau^{k} \neq \tau^{k-1} $}
        \State Compute the Perron root $\lambda^{\tau^{k}}$ and a Perron eigenvector $X^{\tau^{k}}$ of $\Mt{\tau^{k}}$.
        \State Set $\tau^{k+1}(t)= \tau^{k}(t)$ for $t=1,\ldots, |T|$.
	\State  Compute the vector $$(\epsilon_{t})_{1 \leq t \leq |T|}  = (F_t(X^{\tau^k})- \lambda^{\tau^k}X^{\tau^k}_t)_{1 \leq t \leq |T|}.$$ 
	\State Choose $t^{*} \in \text{arg} \max_{1 \leq t \leq |T|} \; \epsilon_{t}.$
	\State Choose $$\tau^{k+1}(t^{*}) \in 
\text{arg} \max_{p\in P,\; (t^{*},p) \in E} \; \sum_{t'\in T,\; (p,t')\in E}m_{p,t'}X^{\tau^{k}}_{t'}\enspace .$$
\State $k \gets k + 1 $

\EndWhile
\State \textbf{return} the first policy $\tau^{k}$ such that $\tau^{k-1} = \tau^{k}$, the Perron root $\lambda^{\tau^k}$ and a Perron eigenvector $X^{\tau^k}$ of the matrix $\Mt{\tau^{k}}$.

\end{algorithmic}
\end{algorithm}

One can also adapt the Algorithms~\ref{algossp} and~\ref{algossp:DF-Dz} to the Tribune-free setting, again using minimization instead of maximization and replacing $\Ftau$ by $F^{\delta}$.

\if{%
TODO
In~\cite{protasov}, the convergence of the Algorithm~\ref{algossp} is proved 
when all the matrices $\Mt{\tau}$, $\tau\in \cP_{T}$, have  positive entries. 
The arguments there also apply to the Dantzig variant. 
It is proved
by using the following properties.

\begin{lemma}\label{lem-ssp}
Assume that all matrices $\Mt{\tau}$, $\tau\in \cP_{T}$,
have  positive entries, and let
$\lambda^{\tau}, X^\tau$ be the Perron root and a Perron eigenvector
of $\Mt{\tau}$. We have
\begin{enumerate}
\item \label{lem-ssp1}Let $\tau, \tau'\in \cP_{T}$, be such that there exists a unique $t \in T$ with $\tau'(t) \neq \tau(t)$. Then $\Ftaup_{t}(X^{\tau}) > \Ftau_{t}(X^{\tau}) \Rightarrow \lambda^{\tau'} > \lambda^{\tau}$.

\item \label{lem-ssp2}
If for each $t \in T$,  $\tau^{*}\in \argmax_{\tau \in \cP_{T}} \Ftau_{t}(X^{\tau^{*}})$, then $\lambda^{\tau^{*}} = \max_{\tau \in \cP_{T}} \lambda^{\tau}$.
\end{enumerate}
\end{lemma}
Indeed, using Point~\ref{lem-ssp1}, one gets that the sequence of spectral radii $(\lambda^{\tau^{k}})_{k \geq 1}$ constructed by the algorithm \ref{algossp} is non-decreasing as long as $\tau^{k}\neq \tau^{k-1}$. Since this sequence belongs to the finite set $\{ \lambda^{\tau} \; | \; \tau \in \cP_{T} \}$, the sequence $(\tau^{k})_{k}$ has to be stationary, so
the algorithm ends.
Point~\ref{lem-ssp2} of Lemma~\ref{lem-ssp} shows that it returns the best policy $\tau^{*}$.}\fi

\subsection{Numerical experiments}

We next report numerical experiments in the case 
of Despot-free in order to compare 
Protasov's spectral simplex algorithm (with and without the improvement of Dantzig' pivot rule) with the multiplicative policy iteration algorithm (Algorithm~\ref{algopo1p}).
In the log-log figures \ref{fig:DF-n} and \ref{fig:DF-m}, these algorithms are respectively named "Policy Iteration", "Spectral Simplex" and "Spectral Simplex-D".

We constructed random Despot-free instances in which $D=T$ has cardinal $n$,
and every coordinate of the operator is of the form $F_t(X)= \max_{1\leq p \leq {m}} \sum_{t'}A^p_{tt'} X_{t'}$, where $(A^p_{tt'})$ is a 3-dimensional tensor whose
entries are independent random variables drawn with the uniform
law on $\{1,\dots,15\}$. Remember that $m$ is an integer which represents the number of possible different actions per state $t$.
All the results shown on the figures are the average made over 30 simulations, they concern the situation in which one of the two parameters $m,n$ is kept constant while the other one increases. The time is given in seconds.
The computations were performed on Matlab R2016a, using an Intel(R) Core(TM) i7-6500 CPU @ 2.59GHz processor with 12,0Go of RAM.

\begin{figure}
\caption{Performance of Algorithms~\ref{algopo1p},~\ref{algossp}, and~\ref{algossp:DF-Dz} for different $n=10,...,500$.}
\label{fig:DF-n}
\includegraphics[scale=0.5]{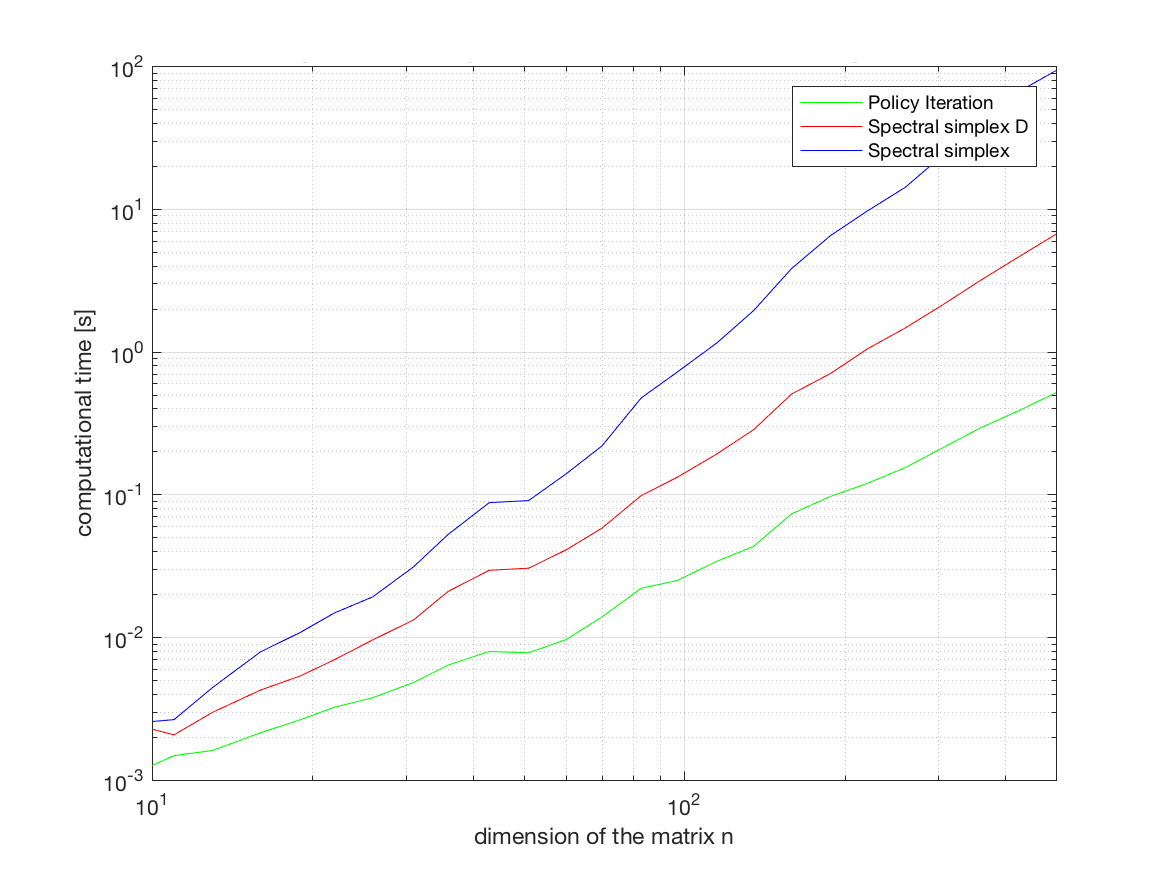}
\end{figure}

\begin{figure}
\caption{Performance of Algorithms~\ref{algopo1p},~\ref{algossp}, and~\ref{algossp:DF-Dz} for different $m=10,...,500$.}
\label{fig:DF-m}
\includegraphics[scale=0.5]{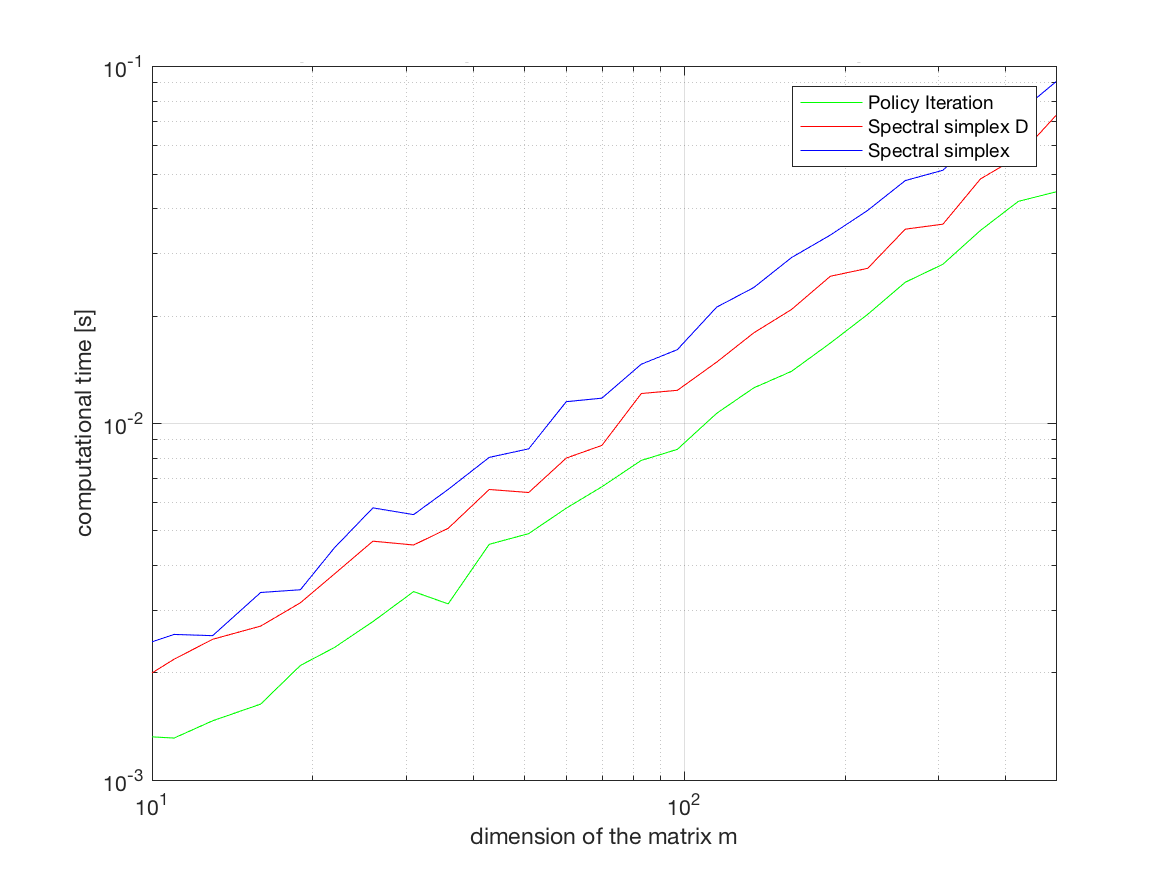}
\end{figure}

In both figures, Spectral Simplex-D appears to be more efficient than the Spectral Simplex algorithm with its original rule. However both algorithms are experimentally outperformed by policy iteration, by one to two order of magnitude, when $n\to\infty$, whereas when $m \to \infty$ (for constant parameter $n$), the performance of the three algorithms seem to deteriorate at the same rate. 
\section{Two-player entropy games: multiplicative policy iteration and power algorithm}\label{algo-jeux}

Let us now consider the general two-player case.
For $\delta\in \cP_D$ and $\tau\in\cP_T$,
let $F^{\delta}$ (resp.\ $\Fdt{\delta}{\tau}$) be the dynamic programming
operator of the game in which the strategy of Despot, $\delta$, is fixed
(resp.\ the strategies of Despot and  of Tribune,
$\delta$ and $\tau$ are fixed), see~\eqref{defFdelta} (resp.\ 
\eqref{defFtaudel}).
We assume here that the matrices 
$\Mdt{\delta}{\tau}$  %
of the linear operators $\Fdt{\delta}{\tau}$, $\delta\in \cP_D,\;
\tau\in \cP_{T}$, see~\eqref{defFtaudel}, are irreducible.

\todo[inline]{MA: je mets le para sur le cas deux joueurs ici}
Then, the Hoffman-Karp's idea~\cite{HoffmanKarp} is readily
adapted to the multiplicative setting: in the 
following algorithm, a sequence $\delta^k$ is constructed
in a similar way as $\tau^k$ in the dual version of Algorithm~\ref{algopo1p},
except that in Step~\ref{step3}, $\lambda^{\delta^{k}}$ and $X^{\delta^{k}}$  are 
computed by applying Algorithm~\ref{algopo1p} 
to the dynamic programming operator $F^{\delta^{k}}$ 
in which the strategy of Despot is fixed to $\delta^k$.
We call this the {\em multiplicative} Hoffman-Karp algorithm.
ta%
It can also be viewed as an ``exact'' version of the 
 policy iteration algorithm of Hoffman 
and Karp~\cite{HoffmanKarp} for $f$.

\begin{algorithm}
\caption{Policy Iteration for two-player entropy games}\label{algopo2p}
\begin{algorithmic}[1]

\State Initialize $k=1$, $\delta^{0}$, $\delta^{1} \neq \delta^{0}$ randomly.

\While{$\delta^{k} \neq \delta^{k-1} $}

\State Apply Algorithm~\ref{algopo1p}  to $F^{\delta^{k}}$. 
This returns a policy $\tau^{k}\in \cP_T$, the Perron root $\lambda^{\delta^{k}}$
and the Perron eigenvector $X^{\delta^k}$ of $\Mdt{\delta^k}{\tau^{k}}$.

\State Compute a new policy $\delta^{k+1}$ such that, for all $d\in D$,

$$ \delta^{k+1}(d) \in
\argmin_{t\in T,\; (d,t)\in E} \; \max_{(t,p)\in E} \left(\sum_{(p,d') \in E} m_{pd'}X^{\delta^{k}}_{d'}\right) \enspace ,$$ %

     taking  $\delta^{k+1}(d) = \delta^{k}(d) $ if it belongs to the latter argmin.

\State $k \gets k + 1 $

\EndWhile

\State \textbf{return} the first policies $(\delta^{k},\tau^{k})$ such that $\delta^{k} = \delta^{k-1}$, 
and the Perron root $\lambda^{\delta^{k}}$
and the Perron eigenvector $X^{\delta^k}$ of $M^{\delta^k,\tau^{k}}$.

\end{algorithmic}
\end{algorithm}

A similar proof as the one of Proposition~\ref{conv-policy} shows that 
Algorithm~\ref{algopo2p}, implemented in exact arithmetics, 
terminates and is correct 
under the previous assumption that for any pair of policies
of the two players,  the associated transition matrix is irreducible. 
Indeed, as for the dual version of Algorithm~\ref{algopo1p},
 $\lambda^{\delta^k}$ is decreasing
as long as $\delta^{k}\neq \delta^{k-1}$.
Then, since again the set of policies of Despot 
is finite, the algorithm ends.

To have an additional point of comparison, we used a power type algorithm,
more precisely,  the projective version of the Krasnoselski-Mann iteration,
proposed in~\cite{stott}. The original Krasnoselski-Mann iteration was developed
in~\cite{mann,krasno}, its analysis was extended and refined by Ishikawa~\cite{ishikawa} and Baillon and Bruck~\cite{baillonbruck}.
Recall that
$$F_{d}(X) = \min_{(d,t)\in E} \max_{(t,p)\in E} \sum_{(p,d') \in E} m_{pd'}X_{d'}\enspace ,$$
and let us define 
$$ H_{d}(X) = (X_{d}G_{d}(X))^{1/2}, \text{ where }G_{d}(X)=\dfrac{F_{d}(X)}{(\prod_{d'\in E}F_{d'}(X))^{1/|D|}}.$$
Every fixed point of $H$ is an eigenvector $X \in R_{+}^{D}$
of $F$ such that $\prod_{d'\in E}X_d=1$, indeed
\begin{align*}
H(X)=X & \iff (X_{d}G_{d}(X))^{1/2}=X_{d}, \forall d \in D \\
& \iff G_{d}(X)=X_{d} , \forall d \in D \\
& \iff \dfrac{F_{d}(X)}{(\prod_{d'}F_{d'}(X))^{1/|D|}} = X_{d}, \forall d \in D \\
&\iff F_{d}(X) = \lambda X_{d},\forall d \in D,\;\text{and}\; \prod_{d'\in E}X_d=1
\enspace,
\end{align*} 
for some $\lambda >0$ (since then $\lambda= (\prod_{d'}F_{d'}(X))^{1/|D|}$).
 It can be shown
as a corollary of a general result of Ishikawa~\cite{ishikawa} concerning
nonexpansive mappings in Banach spaces that Algorithm~\ref{algoPower2p}
does converge if $F$ has an eigenvector in the interior of the cone. We
refer the reader to~\cite{stott} for more details on the analysis
of the projective Krasnoselski-Mann iteration. We use 
Hilbert's projective metric $d_{H}(X,Y) = \| \log(X) - \log(Y) \|_{H}$ (with $\|X\|_{H} = \max_{d} X_{d}- \min_{d} X_{d}$) to test the approximate termination.

\begin{algorithm}
\caption{Power algorithm for two-player entropy games}\label{algoPower2p}
\begin{algorithmic}[1]

\State Initialize $k=0$, $V^{0} = \boldsymbol{e} \in \R^{D}, V^{1} = F(X^{0})$.

\While{$ d_{H}(V^{k+1}, V^{k}) > \epsilon $}
\State $V^{k} \gets V^{k+1} $
\State $V^{k+1} \gets H(V^{k}) $
\State $k \gets k + 1 $
\EndWhile

\State \textbf{return} the first vector $V^{k}$ such that $ d_{H}(V^{k+1}, V^{k}) \leq \epsilon.$

\end{algorithmic}
\end{algorithm}

The log-log figures~\ref{fig:2P-n} and~\ref{fig:2P-m} show the performances of the power algorithm and the two-player policy iteration algorithm. The computation were performed on the same computer as in the previous section;
the time is still given in seconds. The policy iteration algorithm outperforms the power algorithm asymptotically for large number of actions $m$, whereas for large number of states $n$, the power algorithm is more efficient. The experimentally observed efficiency of the --naive-- power algorithm actually reveals
that the random instances we considered are relatively ``easy''. 

\begin{figure}[H]

\caption{Performance of Algorithms~\ref{algopo2p} and~\ref{algoPower2p} for $n=10,...,2000$.}
\label{fig:2P-n}
\includegraphics[scale=0.21]{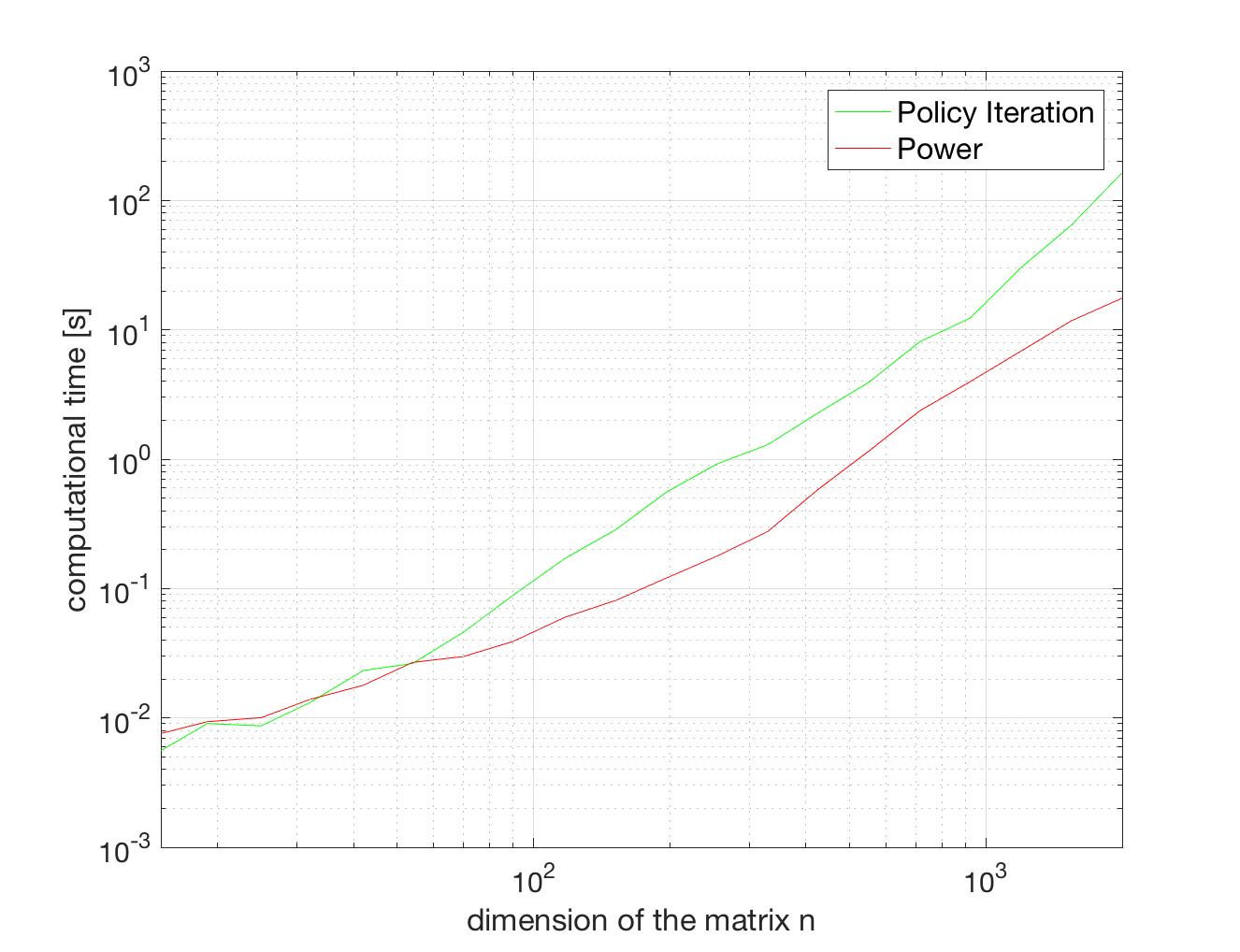}
\end{figure}
\begin{center}

\begin{figure}[H]

\caption{Performance of Algorithms~\ref{algopo2p} and~\ref{algoPower2p} for $m=10,...,200$.}
\label{fig:2P-m}
\includegraphics[scale=0.21]{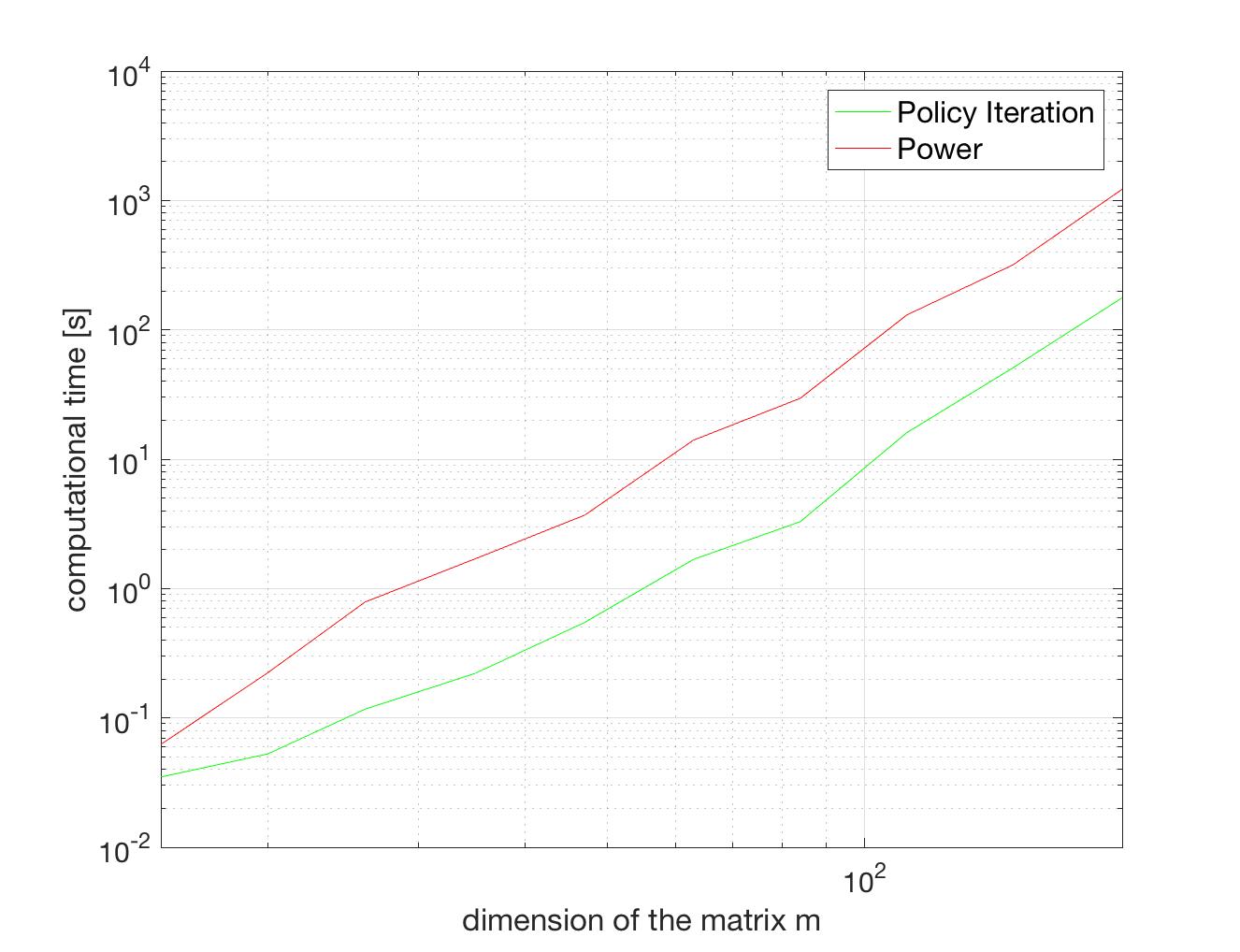}
\end{figure}

\end{center}

\section{Concluding remarks}
We developed an operator approach for entropy games,
relating them with risk sensitive control via non-linear
Perron-Frobenius theory. This leads to 
a theoretical result (polynomial time solvability
of the Despot-free case), and this allows one to adapt
policy iteration to these games. Several issues concerning policy iteration
in the spectral setting remains unsolved. A first issue is to understand
what kind of approximate eigenvalue algorithms are best suited. 
A second issue is to identify significant classes 
of entropy games on which the Hoffman-Karp type policy iteration algorithm
can be shown to run in polynomial time (compare with~\cite{ye2010simplex,hansen2011strategy} in the case of Markov decision processes).
In view of the asymmetry between Despot and Tribune,
one may expect that Tribune-free entropy games are at least as hard as  deterministic mean payoff games, it would be interesting to confirm that this is the case. 

\subparagraph*{Acknowledgments.}
An announcement of the present results appeared in the proceedings of STACS, \cite{akian_et_al:LIPIcs:2017:7026}. We are very grateful to the referees of this STACS paper and also to the referees of the present extended version,
for their detailed comments which helped us to
improve this manuscript.
\bibliographystyle{alpha}
\bibliography{nonexpansive}

\newcommand{\etalchar}[1]{$^{#1}$}
\begin{thebibliography}{AGGCG17}

\bibitem[AB17]{anantharam}
V.~Anantharam and V.~S. Borkar.
\newblock A variational formula for risk-sensitive reward.
\newblock {\em SIAM J. Contro Optim.}, 55(2):961--988, 2017.
\newblock arXiv:1501.00676.

\bibitem[ACD{\etalchar{+}}16]{asarin}
E.~Asarin, J.~Cervelle, A.~Degorre, C.~Dima, F.~Horn, and V.~Kozyakin.
\newblock Entropy games and matrix multiplication games.
\newblock In {\em 33rd Symposium on Theoretical Aspects of Computer Science,
  {STACS} 2016, February 17-20, 2016, Orl{\'{e}}ans, France}, pages
  11:1--11:14, 2016.

\bibitem[AGG12]{AGGut10}
M.~Akian, S.~Gaubert, and A.~Guterman.
\newblock Tropical polyhedra are equivalent to mean payoff games.
\newblock {\em International Journal of Algebra and Computation}, 22(1):125001
  (43 pages), 2012.

\bibitem[AGGCG17]{akian_et_al:LIPIcs:2017:7026}
M.~Akian, S.~Gaubert, J.~Grand-Cl{\'e}ment, and J.~Guillaud.
\newblock {The Operator Approach to Entropy Games}.
\newblock In H.~Vollmer and B.~Vall\'ee, editors, {\em 34th Symposium on
  Theoretical Aspects of Computer Science (STACS 2017)}, volume~66 of {\em
  Leibniz International Proceedings in Informatics (LIPIcs)}, pages 6:1--6:14,
  Dagstuhl, Germany, 2017. Schloss Dagstuhl--Leibniz-Zentrum fuer Informatik.

\bibitem[AGN11]{agn}
M.~Akian, S.~Gaubert, and R.~Nussbaum.
\newblock A {C}ollatz-{W}ielandt characterization of the spectral radius of
  order-preserving homogeneous maps on cones.
\newblock arXiv:1112.5968, 2011.

\bibitem[AM09]{andersson_miltersen}
D.~Andersson and P.B. Miltersen.
\newblock The complexity of solving stochastic games on graphs.
\newblock In {\em Proceedings of ISAAC'09}, number 5878 in LNCS. Springer,
  2009.

\bibitem[BB88]{borwein}
J.~M. Borwein and P.~B. Borwein.
\newblock On the complexity of familiar functions and numbers.
\newblock {\em SIAM Review}, 30(4):589--601, 1988.

\bibitem[BB92]{baillonbruck}
J.~B. Baillon and R.~E. Bruck.
\newblock Optimal rates of asymptotic regularity for averaged nonexpansive
  mappings.
\newblock In K.~K. Tan, editor, {\em Proceedings of the Second International
  Conference on Fixed Point Theory and Applications}, pages 27--66. World
  Scientific Press, 1992.

\bibitem[BGV14]{bolte2013}
J.~Bolte, S.~Gaubert, and G.~Vigeral.
\newblock Definable zero-sum stochastic games.
\newblock {\em Mathematics of Operations Research}, 40(1):171--191, 2014.

\bibitem[BK76]{BewlKohl76}
T.~Bewley and E.~Kohlberg.
\newblock The asymptotic theory of stochastic games.
\newblock {\em Math. Oper. Res.}, 1(3):197--208, 1976.

\bibitem[BN09]{blondel}
V.~D. Blondel and Y.~Nesterov.
\newblock Polynomial-time computation of the joint spectral radius for some
  sets of nonnegative matrices.
\newblock {\em SIAM J. Matrix Anal.}, 31(3):865--876, 2009.

\bibitem[BP94]{berman}
A.~Berman and R.J. Plemmons.
\newblock {\em Nonnegative matrices in the mathematical sciences}.
\newblock Academic Press, 1994.

\bibitem[CH14]{chen}
T.~Chen and T.~Han.
\newblock On the complexity of computing maximum entropy for markovian models.
\newblock In {\em 34th International Conference on Foundation of Software
  Technology and Theoretical Computer Science, {FSTTCS} 2014, December 15-17,
  2014, New Delhi, India}, pages 571--583, 2014.

\bibitem[CT80]{crandall}
M.G. Crandall and L.~Tartar.
\newblock Some relations between non expansive and order preserving maps.
\newblock {\em Proceedings of the AMS}, 78(3):385--390, 1980.

\bibitem[DV75]{donsker}
M.~D. Donsker and S.~R.~S Varadhan.
\newblock On a variational formula for the principal eigenvalue for operators
  with maximum principle.
\newblock {\em Proc. Nat. Acad. Sci. USA}, 72(3):780--783, 1975.

\bibitem[FHH97]{FHH97}
W.~H. Fleming and D.~Hern{\'a}ndez-Hern{\'a}ndez.
\newblock Risk-sensitive control of finite state machines on an infinite
  horizon. {I}.
\newblock {\em SIAM J. Control Optim.}, 35(5):1790--1810, 1997.

\bibitem[FHH99]{FHH99}
W.~H. Fleming and D.~Hern{\'a}ndez-Hern{\'a}ndez.
\newblock Risk-sensitive control of finite state machines on an infinite
  horizon. {II}.
\newblock {\em SIAM J. Control Optim.}, 37(4):1048--1069 (electronic), 1999.

\bibitem[GG98]{gg98a}
S.~Gaubert and J.~Gunawardena.
\newblock A non-linear hierarchy for discrete event dynamical systems.
\newblock In {\em Proc. of the Fourth Workshop on Discrete Event Systems
  (WODES98)}, pages 249--254, Cagliari, Italy, 1998. IEE.

\bibitem[GG04]{arxiv1}
S.~Gaubert and J.~Gunawardena.
\newblock The {P}erron-{F}robenius theorem for homogeneous, monotone functions.
\newblock {\em Trans. of AMS}, 356(12):4931--4950, 2004.

\bibitem[GLS81]{schrijver}
M.~Gr\"otschel, L.~Lov\'asz, and A.~Schrijver.
\newblock The ellipsoid method and its consequences in combinatorial
  optimization.
\newblock {\em Combinatorica}, 1(2):169--197, 1981.

\bibitem[GS18]{stott}
S.~Gaubert and N.~Stott.
\newblock A convergent hierarchy of non-linear eigenproblems to compute the
  joint spectral radius of nonnegative matrices.
\newblock To appear in the proceedings of the 23rd International Symposium on
  Mathematical Theory of Networks and Systems (MTNS2018), Hong Kong, 2018.

\bibitem[GV12]{GV10}
S.~Gaubert and G.~Vigeral.
\newblock A maximin characterization of the escape rate of nonexpansive
  mappings in metrically convex spaces.
\newblock {\em Math. Proc. of Cambridge Phil. Soc.}, 152:341--363, 2012.

\bibitem[HK66]{HoffmanKarp}
A.~J. Hoffman and R.~M. Karp.
\newblock On nonterminating stochastic games.
\newblock {\em Management Science. Journal of the Institute of Management
  Science. Application and Theory Series}, 12:359--370, 1966.

\bibitem[HM72]{Howard-Matheson}
R.~A. Howard and J.~E. Matheson.
\newblock Risk-sensitive markov decision processes.
\newblock {\em Management Science}, 18(7):356--369, 1972.

\bibitem[HMZ11]{hansen2011strategy}
T.D. Hansen, P.B. Miltersen, and U.~Zwick.
\newblock Strategy iteration is strongly polynomial for 2-player turn-based
  stochastic games with a constant discount factor.
\newblock In {\em Innovations in Computer Science 2011}, pages 253--263.
  Tsinghua University Press, 2011.

\bibitem[Ish76]{ishikawa}
Shiro Ishikawa.
\newblock Fixed points and iteration of a nonexpansive mapping in a {B}anach
  space.
\newblock {\em Proceedings of the American Mathematical Society}, 59(1):65--71,
  1976.

\bibitem[Kin61]{kingman}
J.F.C. Kingman.
\newblock A convexity property of positive matrices.
\newblock {\em Quart. J. Math. Oxford, Ser. 2}, 12:283--284, 1961.

\bibitem[Koz15]{kozyakin}
V.~Kozyakin.
\newblock Hourglass alternative and the finiteness conjecture for the spectral
  characteristics of sets of non-negative matrices.
\newblock arXiv:1507.00492, 2015.

\bibitem[Kra55]{krasno}
M.~A. Krasnosel'ski\u{i}.
\newblock Two remarks on the method of successive approximations.
\newblock {\em Uspekhi Matematicheskikh Nauk}, 10:123--127, 1955.

\bibitem[Kul97]{kullback}
Solomon Kullback.
\newblock {\em Information theory and statistics}.
\newblock Dover Publications, Inc., Mineola, NY, 1997.
\newblock Reprint of the second (1968) edition.

\bibitem[LLNW18]{lemmens}
B.~Lemmens, B.~Lins, R.~Nussbaum, and M.~Wortel.
\newblock Denjoy-{W}olff theorems for {H}ilbert’s and {T}hompson’s metric
  spaces.
\newblock {\em Journal d'Analyse Mathématique}, 134:671--718, 2018.

\bibitem[Lot05]{lothaire}
M.~Lothaire.
\newblock {\em Applied Combinatorics on Words}.
\newblock Cambridge, 2005.

\bibitem[Man53]{mann}
W.~R. Mann.
\newblock Mean value methods in iteration.
\newblock {\em Proceedings of the American Mathematical Society}, 4:506--510,
  1953.

\bibitem[MN81]{MertNeym81}
J.-F. Mertens and A.~Neyman.
\newblock Stochastic games.
\newblock {\em Internat. J. Game Theory}, 10(2):53--66, 1981.

\bibitem[M{\"u}l05]{muller}
J.~M. M{\"u}ller.
\newblock {\em ELementary functions: algorithms and implementation}.
\newblock Birkha\"user, 2005.

\bibitem[Ney03]{neymansurv}
A.~Neyman.
\newblock Stochastic games and nonexpansive maps.
\newblock In {\em Stochastic games and applications (Stony Brook, NY, 1999)},
  volume 570 of {\em NATO Sci. Ser. C Math. Phys. Sci.}, pages 397--415. Kluwer
  Acad. Publ., Dordrecht, 2003.

\bibitem[Nus86]{nussbaum86}
R.~D. Nussbaum.
\newblock Convexity and log convexity for the spectral radius.
\newblock {\em Linear Algebra Appl.}, 73:59--122, 1986.

\bibitem[Pro15]{protasov}
V.~Yu. Protasov.
\newblock Spectral simplex method.
\newblock {\em Mathematical Programming}, 2015.

\bibitem[Put05]{putermanbook}
M.~L. Puterman.
\newblock {\em Markov Decision Processes}.
\newblock Wiley, 2005.

\bibitem[Rot84]{rothblum}
U.~G. Rothblum.
\newblock Multiplicative markov decision chains.
\newblock {\em Mathematics of Operations Research}, 9(1):6--24, 1984.

\bibitem[Rum79]{rump}
S.M. Rump.
\newblock Polynomial minimum root separation.
\newblock {\em Mathematics of Computation}, 145(33):327--336, 1979.

\bibitem[RW98]{RockafellarWets}
R.~T. Rockafellar and R.~J.-B. Wets.
\newblock {\em Variational analysis}.
\newblock Springer-Verlag, Berlin, 1998.

\bibitem[Sla76]{Sladky1976}
K.~Sladk{\'y}.
\newblock {\em On dynamic programming recursions for multiplicative Markov
  decision chains}, pages 216--226.
\newblock Springer Berlin Heidelberg, Berlin, Heidelberg, 1976.

\bibitem[vdD98]{Dries98}
L.~van~den Dries.
\newblock {\em Tame topology and o-minimal structures}, volume 248 of {\em
  London Mathematical Society Lecture Note Series}.
\newblock Cambridge University Press, Cambridge, 1998.

\bibitem[vdD99]{vandendriessurvey}
L.~van~den Dries.
\newblock o-minimal structures and real analytic geometry.
\newblock In {\em Current developments in mathematics, 1998 ({C}ambridge,
  {MA})}, pages 105--152. Int. Press, Somerville, MA, 1999.

\bibitem[Vig13]{Vipreprint}
G.~Vigeral.
\newblock A zero-sum stochastic game with compact action sets and no asymptotic
  value.
\newblock {\em Dynamic Games and Applications}, 3(2):172--186, 2013.

\bibitem[Whi82]{whittle}
P.~Whittle.
\newblock {\em Optimization over time, I}.
\newblock Wiley, 1982.

\bibitem[Wil96]{wilkie}
A.~J. Wilkie.
\newblock Model completeness results for expansions of the ordered field of
  real numbers by restricted {P}faffian functions and the exponential function.
\newblock {\em J. Amer. Math. Soc.}, 9(4):1051--1094, 1996.

\bibitem[Ye11]{ye2010simplex}
Y.~Ye.
\newblock The simplex and policy-iteration methods are strongly polynomial for
  the markov decision problem with a fixed discount rate.
\newblock {\em Mathematics of Operations Research}, 36(4):593--603, 2011.

\bibitem[Zij87]{zijmjota}
W.~H.~M. Zijm.
\newblock Asymptotic expansions for dynamic programming recursions with general
  nonnegative matrices.
\newblock {\em J. Optim. Theory Appl.}, 54(1):157--191, 1987.

\end{thebibliography}

\end{document}